\pdfoutput=1
\documentclass[draftclsnofoot, onecolumn, 12pt]{IEEEtran}
\usepackage{mathtools} 
\usepackage{amsthm, amssymb} 
\usepackage{cancel}
\usepackage{bm} 
\usepackage{microtype} 
\usepackage{booktabs}
\usepackage{siunitx} 
\usepackage{pgfplots}
\usepackage[noadjust]{cite}
\usepackage[hidelinks]{hyperref} 
\usepackage[nameinlink]{cleveref} 
\usepackage{orcidlink} 

\usetikzlibrary{calc}
\hypersetup{hidelinks, pdfinfo = {Title = {Multi-Access Distributed Computing}, Author = {Federico Brunero and Petros Elia}, Subject = {Information-Theoretic Computing}, Keywords = {Coded Distributed Computing, Coded Multicasting, Communication Load, Information-Theoretic Converse, MapReduce, Multi-Access Distributed Computing (MADC)}}}

\allowdisplaybreaks
\pgfplotsset{compat = newest}

\Crefformat{figure}{#2Fig.~#1#3}

\theoremstyle{plain}
\newtheorem{theorem}{Theorem}
\newtheorem{lemma}{Lemma}
\newtheorem{corollary}{Corollary}[theorem]

\theoremstyle{definition}
\newtheorem{definition}{Definition}

\theoremstyle{remark}
\newtheorem{remark}{Remark}

\newcommand{\refappendix}[1]{\hyperref[#1]{Appendix~\ref*{#1}}}

\title{Multi-Access Distributed Computing}

\author{Federico Brunero\textsuperscript{\orcidlink{0000-0002-6980-3827}} and Petros Elia\textsuperscript{\orcidlink{0000-0002-3531-120X}}%
    \thanks{%
    This work was supported by the European Research Council (ERC) through the EU Horizon 2020 Research and Innovation Program under Grant 725929 (Project DUALITY).}
    \thanks{%
    The authors are with the Communication Systems Department at EURECOM, 450 Route des Chappes, 06410 Sophia Antipolis, France (email: brunero@eurecom.fr; elia@eurecom.fr).
  }
}

\begin{document}

\maketitle

\begin{abstract}
    Coded distributed computing (CDC) is a new technique proposed with the purpose of decreasing the intense data exchange required for parallelizing distributed computing systems. Under the famous MapReduce paradigm, this coded approach has been shown to decrease this communication overhead by a factor that is linearly proportional to the overall computation load during the mapping phase. Nevertheless, it is widely accepted that this overhead remains a main bottleneck in distributed computing. To address this, we take a new approach and we explore a new system model which, for the same aforementioned overall computation load of the mapping phase, manages to provide astounding reductions of the communication overhead and, perhaps counterintuitively, a substantial increase of the computational parallelization. In particular, we propose multi-access distributed computing (MADC) as a novel generalization of the original CDC model, where now \emph{mappers} (nodes in charge of the map functions) and \emph{reducers} (nodes in charge of the reduce functions) are distinct computing nodes that are connected through a multi-access network topology. Focusing on the MADC setting with combinatorial topology, which implies $\Lambda$ mappers and $K$ reducers such that there is a unique reducer connected to any $\alpha$ mappers, we propose a novel coded scheme and a novel information-theoretic converse, which jointly identify the optimal inter-reducer communication load, as a function of the computation load, to within a constant gap of $1.5$. Additionally, a modified coded scheme and converse identify the optimal max-link communication load across all existing links to within a gap of $4$. The unparalleled coding gains reported here should not be simply credited to having access to more mapped data, but rather to the powerful role of topology in effectively aligning mapping outputs. This realization raises the open question of which multi-access network topology guarantees the best possible performance in distributed computing.
\end{abstract}


\section{Introduction}

With the development of large-scale machine learning algorithms and applications relying heavily on large volumes of data, we are now experiencing an ever-growing need to distribute large computations across multiple computing nodes. Different computing frameworks, such as MapReduce~\cite{Dean2008MapReduceSimplifiedData} and Spark~\cite{Zaharia2010SparkClusterComputing}, have been proposed to address these needs, based on the aforementioned simple yet powerful concept of distributing large-scale algorithms --- to be executed over a set of input data files --- across multiple computing machines. Under the well-known MapReduce framework, the overall process is typically split in three distinct phases, starting with the \emph{map phase}, the \emph{shuffle phase} and then the \emph{reduce phase}. During the map phase, each computing node is assigned a subset of the input data files, and proceeds to apply to each locally available file certain designated map functions. The outputs of such map functions, referred to as intermediate values (IVs), are then exchanged among the computing nodes during the shuffle phase, so that each computing node can retrieve any missing, required IVs it did not compute locally. Finally, during the reduce phase, each computing node computes one (or more) output functions depending on its assigned reduce functions, each of which takes as input the IVs computed for each input file.

\subsection{Coded Distributed Computing}

Several studies have shown that the aforementioned distributed map-shuffle-reduce approach comes with bottlenecks that may severely hinder the parallelization of computationally-intensive operations. While some works~\cite{Lee2018SpeedingDistributedMachine, Kumar2014ComprehensiveReviewStraggler} focused on the impact of straggler nodes, other works have pointed out that the total execution time of a distributed computing application is often dominated by the shuffling process. For instance, the work in~\cite{Zhuoyao2013PerformanceModelingMapReduce}, having explored the behavior of several algorithms on the Amazon EC2 cluster, revealed that the communication load in the shuffle phase was in fact the dominant bottleneck in computing the above tasks in a distributed manner. Similarly, the authors in~\cite{Li2018FundamentalTradeoffComputation} observed that, for the execution of a conventional TeraSort application, more than $\SI{95}{\percent}$ of the overall execution time was spent for inter-node communication.

Motivated by this communication bottleneck in the shuffle phase, the authors in~\cite{Li2018FundamentalTradeoffComputation} introduced coded distributed computing (CDC) as a novel framework that can yield lower communication loads during data shuffling. This gain could be attributed to a careful and joint design of the map and the shuffle phases. Approaching the distributed computing problem from an information-theoretic perspective, the authors brought to light the interesting relationship between the computation load during the mapping phase, and the communication load of the shuffling step. In particular, the work in~\cite{Li2018FundamentalTradeoffComputation} revealed that if the computation load of the mapping phase is \emph{carefully} increased by a factor $r$ --- which means that each input file is mapped on average by $r$ \emph{carefully chosen} computing nodes --- then the communication load can be reduced by the same factor $r$ by employing coding techniques during the shuffle phase\footnote{This speedup factor $r$ is often referred to as the \emph{coding gain}, and it reflects the number of computing servers that simultaneously benefit from a single transmission.}. Building on the coding-based results in cache networks~\cite{MaddahAli2014FundamentalLimitsCaching, Ji2016FundamentalLimitsCaching}, the work in~\cite{Li2018FundamentalTradeoffComputation} characterized the exact information-theoretic tradeoff between this computation and communication loads under any map-shuffle-reduce scheme with uniform mapping capabilities and uniform assignment of reduce functions. 

Since its original information-theoretic inception in~\cite{Li2018FundamentalTradeoffComputation}, coded distributed computing has been explored with several variations. Such variations include heterogeneity aspects where, for example, each computing node may be assigned different numbers of files to be mapped and functions to be reduced. For such settings, novel schemes, based on hypercube and hypercuboid geometric structures, were developed in~\cite{Woolsey2021NewCombinatorialCoded, Woolsey2021CombinatorialDesignCascaded}, which managed not only to compensate for the heterogeneous nature of the considered scenarios, but to also exploit these asymmetries in order to require a smaller number of input files\footnote{It is worth noting here the importance of designing schemes that can work with a smaller number of input files. The coded scheme in~\cite{Li2018FundamentalTradeoffComputation}, albeit achieving the information-theoretic optimal, requires a number of input files that increases exponentially with the number of computing nodes. This may entail some limitations when finite-sized data sets are considered. Finding schemes with good performance and low file-number requirements is a research direction of significant importance.}, compared to the initial scheme in~\cite{Li2018FundamentalTradeoffComputation}. Regarding this problem of requiring a large number of input files, it is worth also mentioning the work in~\cite{Parrinello2018CodedDistributedComputing}, where the authors proposed a system model for distributed computing, where the required number of input files was lowered dramatically under an assumption of a multi-rank wireless network.

Some additional works explored the scenario where the computing servers communicate with each other through switch networks~\cite{Wan2020TopologicalCodedDistributed} or in the presence of a randomized connectivity~\cite{Srinivasavaradhan2018DistributedComputingTrade}, whereas some other works further investigated distributed computing over wireless channels~\cite{Li2017ScalableFrameworkWireless}, as well as explored the interesting scenario where each computing node might have limited storage and computational resources~\cite{Yan2018StorageComputationCommunication, Yan2018StorageComputationCommunicationa}. A comprehensive survey on CDC is nicely presented in~\cite{Ng2021ComprehensiveSurveyCoded}.

\subsection{Contributions}

In this work, we propose the new multi-access distributed computing (MADC) model, which can be considered as an extension of the original setting introduced in~\cite{Li2018FundamentalTradeoffComputation}, and which entails \emph{mappers} (map nodes) being connected to various \emph{reducers} (reduce nodes), and where these mappers and reducers are now distinct entities\footnote{This choice is reasonable if we think of mappers as computing nodes that are specialized in evaluating the map functions, and of reducers as computing nodes that are specialized in evaluating the reduce functions~\cite{Shan2010FpmrMapreduceFramework, Choi2014MapReduceProcessing}.}. As is common, mappers are in charge of mapping subsets of the input files, whereas the reducers are in charge of collecting the IVs in order to compute the reduce functions. We will here focus on the so-called \emph{combinatorial topology} which will define how the mappers are connected to the reducers. This is a widely studied topology in other settings outside of distributed computing~\cite{Muralidhar2021MaddahAliNiesen, Brunero2021FundamentalLimitsCombinatorial, Vaidya2022MultiAccessCache}, and it will --- as we will discover later on --- allow for stunning gains. Under such combinatorial topology, we consider $\Lambda$ map nodes and $K \geq \Lambda$ reduce nodes, where each map node maps a subset of the input files, where each reduce node is connected to $\alpha$ map nodes, and where there is exactly one reducer for each subset of $\alpha$ map nodes. Each reducer can retrieve intermediate values only from the mappers it is connected to, whereas these same reducers can exchange via an error-free shared-link broadcast channel the remaining required intermediate values. A simple schematic of the model is shown in \Cref{fig: Multi-Access Distributed Computing Example} for the case $\Lambda = 4$, $\alpha = 2$ and $K = 6$.

\begin{figure}
    \centering
    \tikzset{cnode/.style={draw, circle, inner sep = 0, minimum size = 1cm}}
    \tikzset{snode/.style={draw, rectangle, text centered}}
    \begin{tikzpicture}
        \draw (-5, 0)--(5, 0);
        \foreach \a/\b/\c/\n/\o in {1/1/2/-5/black, 2/1/3/-3/red, 3/1/4/-1/blue, 4/2/3/1/orange, 5/2/4/3/purple, 6/3/4/5/teal}{
        \draw (\n, 0)--++(0, -0.75) node[cnode, below, \o](reducer\a){$\b\c$};
        }
        \node at (reducer1.west)[left, xshift = -0.125cm]{Reducers};
        \node at (0, 0.25)[above]{Broadcast Channel};
        \node at ($(reducer1.south) + (1.25, -0.75)$)[snode, below, minimum height = 0.75cm, minimum width = 1cm](mapper1){$1$};
        \node at ($(reducer1.south) + (3.75, -0.75)$)[snode, below, minimum height = 0.75cm, minimum width = 1cm](mapper2){$2$};
        \node at ($(reducer1.south) + (6.25, -0.75)$)[snode, below, minimum height = 0.75cm, minimum width = 1cm](mapper3){$3$};
        \node at ($(reducer1.south) + (8.75, -0.75)$)[snode, below, minimum height = 0.75cm, minimum width = 1cm](mapper4){$4$};
        \draw[semithick] (reducer1.south)--(mapper1.north); \draw[semithick] (reducer1.south)--(mapper2.north);
        \draw[semithick, red] (reducer2.south)--(mapper1.north); \draw[semithick, red] (reducer2.south)--(mapper3.north);
        \draw[semithick, blue] (reducer3.south)--(mapper1.north); \draw[semithick, blue] (reducer3.south)--(mapper4.north);
        \draw[semithick, orange] (reducer4.south)--(mapper2.north); \draw[semithick, orange] (reducer4.south)--(mapper3.north);
        \draw[semithick, purple] (reducer5.south)--(mapper2.north); \draw[semithick, purple] (reducer5.south)--(mapper4.north);
        \draw[semithick, teal] (reducer6.south)--(mapper3.north); \draw[semithick, teal] (reducer6.south)--(mapper4.north);
        \node at (mapper1.west)[left, xshift = -0.125cm]{Mappers};
    \end{tikzpicture}
    \caption{Multi-access distributed computing problem with $\Lambda = 4$ mappers and $K = 6$ reducers, where each reducer is connected exactly and uniquely to a subset of $\alpha = 2$ map nodes.}
    \label{fig: Multi-Access Distributed Computing Example}
\end{figure}
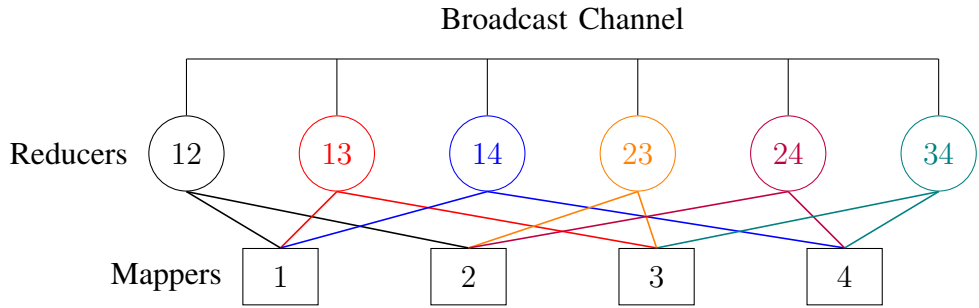

As discussed before, the communication load in the shuffle phase can represent a significant bottleneck of distributed computing. As a consequence, our objective is to minimize the volume of data exchanged by the reducers over the common-bus link during the shuffle phase, as well as the communication load between the mappers and the reducers. We start our analysis by first neglecting the communication cost between mappers and reducers, and we propose --- for the aforementioned MADC model with combinatorial topology --- a novel coded scheme that allows for efficient communication over the broadcast communication channel. For such setting, we also provide an information-theoretic lower bound on the communication load, and we show this to be within a constant multiplicative gap of $1.5$ from the achievable communication load guaranteed by the proposed coded scheme. We then proceed to also account for the download cost from mappers to reducers. For such setting, our goal is to minimize the maximal (normalized) number of bits across all links in the system. To this purpose, we introduce an additional mappers-to-reducers communication scheme and a novel converse bound which, together with the previous inter-reducer scheme, allow us to characterize the optimal max-link communication load within a constant multiplicative gap of $4$.

As suggested above, the newly derived fundamental limits suggest outstanding performance. While for any given computation load $r$ the original setting in~\cite{Li2018FundamentalTradeoffComputation} accepts a maximal coding gain of $r$, we here show that the new MADC model with combinatorial topology allows for a coding gain equal to $\left( \binom{r + \alpha}{r} - 1 \right)$, again for the same mapping cost $r$. This we believe is the first time that topology is shown to have such powerful impact in the setting of coded distributed computing.

\subsection{Paper Outline}

The rest of the paper is organized as follows. First, the system model is presented in \Cref{sec: System Model}. Next, \Cref{sec: Main Results} provides the main contributions of the paper. An illustrative example of the novel coded scheme for multi-access distributed computing is then described in \Cref{sec: Illustrative Example}. After that, the general proofs of the achievable schemes and the converse bounds are presented from \Cref{sec: Achievability Without Download Cost Proof} to \Cref{sec: Converse With Download Cost Proof}. Finally, \Cref{sec: Conclusions} concludes the paper. Some additional proofs are provided in the appendices.

\subsection{Notation}

We denote by $\mathbb{N}$ the set of non-negative integers and by $\mathbb{N}^{+}$ the set of positive integers. For $n \in \mathbb{N}^{+}$, we define $[n] \coloneqq \{1, 2, \dots, n\}$. If $a, b \in \mathbb{N}^{+}$ such that $a < b$, then $[a : b] \coloneqq \{a, a + 1, \dots, b - 1, b\}$. For sets we use calligraphic symbols, whereas for vectors we use bold symbols. Given a finite set $\mathcal{A}$, we denote by $|\mathcal{A}|$ its cardinality. For $n, m \in \mathbb{N}^{+}$, we define $[n]_{m} \coloneqq \{\mathcal{A} : \mathcal{A} \subseteq [n], |\mathcal{A}| = m\}$. We denote by $\binom{n}{k}$ the binomial coefficient and we let $\binom{n}{k} = 0$ whenever $n < 0$, $k < 0$ or $n < k$. For $n, m \in \mathbb{N}^{+}$, we denote by $\mathbb{F}^n_{2^m}$ the $n$-dimensional vector space over the finite field with cardinality $2^m$. For $n \in \mathbb{N}^{+}$, we denote by $S_n$ the group of all permutations of $[n]$.

\section{System Model}\label{sec: System Model}

The general distributed computing problem~\cite{Li2018FundamentalTradeoffComputation} consists of computing $Q$ output functions from $N$ input files with $Q, N \in \mathbb{N}^{+}$. Each file $w_n \in \mathbb{F}_{2^F}$ with $n \in [N]$ consists of $F$ bits for some $F \in \mathbb{N}^{+}$, and the $q$-th function is defined as
\begin{equation}
    \phi_q \colon \mathbb{F}^{N}_{2^F} \to \mathbb{F}_{2^B}
\end{equation}
for each $q \in [Q]$, i.e., each function maps all the $N$ input files into a stream $u_q = \phi_q(w_1, \dots, w_N) \in \mathbb{F}_{2^B}$ of $B$ bits. The main assumption is that each function $\phi_{q}$ is \emph{decomposable} and so can be written as
\begin{equation}
    \phi_q(w_1, \dots, w_N) = h_q(g_{q, 1}(w_1), \dots, g_{q, N}(w_N)), \quad \forall q \in [Q]
\end{equation}
where there is a map function $g_{q, n} \colon \mathbb{F}_{2^F} \to \mathbb{F}_{2^T}$ for each $n \in [N]$, which maps the input file $w_n$ into an intermediate value (IV) $v_{q, n} = g_{q, n}(w_n) \in \mathbb{F}_{2^T}$ of $T$ bits, and a reduce function $h_q \colon \mathbb{F}^{N}_{2^T} \to \mathbb{F}_{2^B}$, which maps all the IVs (one per input file) into the output value $u_q = h_q(v_{q, 1}, \dots, v_{q, N}) \in \mathbb{F}_{2^B}$ of $B$ bits.

In this paper, we assume to have machines that are devoted to computing map functions, and machines that are devoted to computing reduce functions. Thus, a node assigned map functions is not assigned reduce functions, and vice versa. In our setting, we consider $\Lambda$ mappers and $K = \binom{\Lambda}{\alpha}$ reducers, where --- in accordance with the combinatorial topology of choice --- there is a unique reducer connected to each subset of $\alpha$ mappers. Denoting by $\mathcal{U} \in [\Lambda]_{\alpha}$ the reducer connected to the $\alpha$ mappers in the set $\mathcal{U}$, before the computation begins, each reducer $\mathcal{U} \in [\Lambda]_{\alpha}$ is assigned a subset $\mathcal{W}_{\mathcal{U}} \subseteq [Q]$ of the output functions, where here $\mathcal{W}_{\mathcal{U}}$ contains the indices of the functions assigned to reducer $\mathcal{U}$. For simplicity, we assume in our setting a symmetric and uniform task assignment, which implies $|\mathcal{W}_{\mathcal{U}}| = Q/K = \eta_2$ for some $\eta_2 \in \mathbb{N}^{+}$ and for each $\mathcal{U} \in [\Lambda]_\alpha$, and $\mathcal{W}_{\mathcal{U}_1} \cap \mathcal{W}_{\mathcal{U}_2} = \emptyset$ for all $\mathcal{U}_1, \mathcal{U}_2 \in [\Lambda]_\alpha$ such that $\mathcal{U}_1 \neq \mathcal{U}_2$. Afterwards, the computation is performed across the set of mappers and reducers in a distributed manner following the map-shuffle-reduce paradigm.

\begin{table}[!htbp]
    \centering
    \renewcommand{\arraystretch}{1.3}
    \caption{Important parameters for the MADC system with combinatorial topology}
    \begin{tabular}{lc@{\hspace{1cm}}lc}
    \toprule
    Number of Mappers & $\Lambda$ & Number of Input Files & $N$ \\
    \midrule
    Multi-Access Degree & $\alpha$ & Communication Load & $L$ \\
    \midrule
    Number of Reducers & $K = \binom{\Lambda}{\alpha}$ & Download Cost & $J$ \\
    \midrule
    Computation Load & $r$ & Max-Link Load & $L_{\textnormal{max-link}} = \max(L, J)$ \\
    \bottomrule
    \end{tabular}
\end{table}

During the map phase, a set of files $\mathcal{M}_\lambda \subseteq \{w_1, \dots, w_N\}$ is assigned to the mapper $\lambda$ for each $\lambda \in [\Lambda]$. Each mapper $\lambda \in [\Lambda]$ computes the intermediate values $\mathcal{V}_{\lambda} = \{ v_{q, n} : q \in [Q], w_n \in \mathcal{M}_\lambda \}$ for all the $Q$ reduce functions using the files in $\mathcal{M}_\lambda$ which it has been assigned. Since reducer $\mathcal{U} \in [\Lambda]_\alpha$ is connected to the mappers in $\mathcal{U}$, it can access the intermediate values in the set $\mathcal{V}_{\mathcal{U}} = \{ v_{q, n} : q \in [Q], w_n \in \mathcal{M}_\mathcal{U} \}$, where $\mathcal{M}_{\mathcal{U}} = \bigcup_{\lambda \in \mathcal{U}} \mathcal{M}_\lambda$ is simply the union set of files assigned to and mapped by the map nodes in $\mathcal{U}$. When the communication cost between mappers and reducers is not neglected, we can define the \emph{download cost} as follows.

\begin{definition}[Download Cost]
    The \emph{download cost}, denoted by $J$, is defined as the maximal normalized number of bits transmitted across the links from the mappers to the reducers, and is given by
    \begin{equation}
        J \coloneqq \max_{\lambda \in [\Lambda]} \max_{\mathcal{U} \in [\Lambda]_{\alpha}: \lambda \in \mathcal{U}} \frac{R^{\mathcal{U}}_{\lambda}}{QNT}
    \end{equation}
    where $R_{\lambda}^{\mathcal{U}}$ denotes the number of bits that are transmitted by mapper $\lambda \in [\Lambda]$ to reducer $\mathcal{U} \in [\Lambda]_{\alpha}$ where $\lambda \in \mathcal{U}$.
\end{definition}

Assuming that each mapper computes all possible IVs from locally available files\footnote{This means that each mapper $\lambda \in [\Lambda]$ computes the intermediate value $v_{q, n}$ for each $q \in [Q]$ and for each $w_n \in \mathcal{M}_\lambda$.}, we define the \emph{computation load} as follows.

\begin{definition}[Computation Load]
    The \emph{computation load}, denoted by $r$, is defined as the total number of files mapped across the $\Lambda$ map nodes and normalized by the total number of files $N$, and it takes the form
    \begin{equation}
        r \coloneqq \frac{\sum_{\lambda \in [\Lambda]} \left|\mathcal{M}_\lambda \right|}{N}.
    \end{equation}
\end{definition}

\begin{remark}
    We wish to point out that the definition of computation load above reflects the overall mapping cost across the $\Lambda$ map nodes. As it will become clear later, our novel system model will allow for a massive computational amelioration (in the reduce phase) --- with a bounded communication overhead --- at the cost of a modest mapping cost across the $\Lambda$ mappers.
\end{remark}

During the shuffle phase, each reducer $\mathcal{U} \in [\Lambda]_\alpha$ retrieves the IVs from the mappers in $\mathcal{U}$ and creates a signal $X_{\mathcal{U}} \in \mathbb{F}_{2^{\ell_\mathcal{U}}}$ for some $\ell_\mathcal{U} \in \mathbb{N}^{+}$ and for some encoding function $\psi_{\mathcal{U}} \colon \mathbb{F}^{Q|\mathcal{M}_{\mathcal{U}}|}_{2^T} \to \mathbb{F}_{2^{\ell_{\mathcal{U}}}}$, where $X_{\mathcal{U}}$ takes the form
\begin{equation}
    X_{\mathcal{U}} = \psi_{\mathcal{U}}(\mathcal{V}_{\mathcal{U}}).
\end{equation}
Then, the signal $X_{\mathcal{U}}$ is multicasted to all other reducers via the broadcast link which connects the reducers. Since such link is assumed to be error-free, each reducer receives all the multicast transmissions without errors. The amount of information exchanged during this phase is referred to as the \emph{communication load}, which is formally defined in the following.

\begin{definition}[Communication Load]
    The \emph{communication load}, denoted by $L$, is defined as the total number of bits transmitted by the $K$ reducers over the broadcast channel during the shuffle phase, and --- after normalization by the number of bits of all intermediate values --- this load is given by
    \begin{equation}
        L \coloneqq \frac{\sum_{\mathcal{U} \in [\Lambda]_\alpha} \ell_{\mathcal{U}}}{QNT}.
    \end{equation}
\end{definition}

Recalling that reducer $\mathcal{U} \in [\Lambda]_{\alpha}$ is assigned a subset of output functions whose indices are in $\mathcal{W}_{\mathcal{U}}$, each reducer $\mathcal{U} \in [\Lambda]_{\alpha}$ wishes to recover the IVs $\{v_{q, n} : q \in \mathcal{W}_{\mathcal{U}}, n \in [N]\}$ to correctly compute $u_q$ for each $q \in \mathcal{W}_{\mathcal{U}}$. Thus, during the reduce phase, each reducer $\mathcal{U} \in [\Lambda]_\alpha$ reconstructs all the needed intermediate values for each $q \in \mathcal{W}_{\mathcal{U}}$ using the messages communicated in the shuffle phase and the intermediate values $\mathcal{V}_{\mathcal{U}}$ retrieved from the mappers in $\mathcal{U}$, i.e., each reducer $\mathcal{U} \in [\Lambda]_\alpha$ computes
\begin{equation}
    (v_{q, 1}, \dots, v_{q, N}) = \chi^q_{\mathcal{U}}(X_{\mathcal{S}} : \mathcal{S} \in [\Lambda]_{\alpha}, \mathcal{V}_{\mathcal{U}})
\end{equation}
for each $q \in \mathcal{W}_{\mathcal{U}}$ and for some decoding function $\chi^q_{\mathcal{U}} \colon \prod_{\mathcal{S} \in [\Lambda]_{\alpha}} \mathbb{F}_{2^{\ell_{\mathcal{S}}}} \times \mathbb{F}^{Q|\mathcal{M}_{\mathcal{U}}|}_{2^T} \to \mathbb{F}^{N}_{2^T}$. In the end, each reducer $\mathcal{U} \in [\Lambda]_\alpha$ computes the output function $u_q = h_q(v_{q, 1}, \dots, v_{q, N})$ for each assigned $q \in \mathcal{W}_{\mathcal{U}}$.

When the download cost is neglected, our goal is to characterize the optimal tradeoff between computation and communication $L^{\star}(r)$. This optimal tradeoff is simply defined as
\begin{equation}
    L^{\star}(r) \coloneqq \inf\{L : \text{$(r, L)$ is achievable}\}
\end{equation}
where the tuple $(r, L)$ is said to be \emph{achievable} if there exists a map-shuffle-reduce procedure such that a communication load $L$ can be guaranteed for a given computation load $r$. On the other hand, when we indeed jointly consider both the inter-reducer communication cost and the mapper-to-reducer download cost, then our aim will be to characterize the optimal max-link communication load $L^{\star}_{\textnormal{max-link}}(r)$, which is defined as
\begin{equation}
    L^{\star}_{\textnormal{max-link}}(r) \coloneqq \inf\{L_{\textnormal{max-link}} : \text{$(r, L_{\textnormal{max-link}})$ is achievable}\}
\end{equation}
where $L_{\textnormal{max-link}} \coloneqq \max(L, J)$ represents the maximum between the communication load and the download cost for a given computation load $r$. In simple words, $L_{\textnormal{max-link}}$ represents the maximal normalized number of bits flowing across any link in the considered system model. Notice that we will assume, throughout the paper, uniform computational capabilities across the mappers and uniform assignment of reduce functions across the reducers, as is commonly assumed (see for example the original work in~\cite{Li2018FundamentalTradeoffComputation}).

\begin{remark}\label{rem: Special Case}
    When $\alpha = 1$, there are $K = \Lambda$ mapper-reducer pairs. If we consider each pair to be a single computing server (which can automatically imply a zero download cost), the proposed system model trivially coincides with the original setting in~\cite{Li2018FundamentalTradeoffComputation}. Hence, since the results in this paper will hold for any $\alpha \in [\Lambda]$, the proposed model can in fact be considered as a proper extension of the original coded distributed computing model.
\end{remark}

\section{Main Results}\label{sec: Main Results}

In this section we will provide our main contributions. As we have already mentioned, we will first consider a setting where the download cost is neglected. Subsequently, we will provide some additional results for the more realistic scenario where the cost of delivering data from the mappers to the reducers is non-zero.

\subsection{Characterizing the Communication Load}\label{sec: Characterizing the Communication Load}

The first result that we provide is the achievable computation-communication tradeoff provided by the novel coded scheme that will be presented in its general form in \Cref{sec: Achievability Without Download Cost Proof}. The result is formally stated in the following theorem.

\begin{theorem}[Achievable Bound]\label{thm: Achievable Bound Without Download Cost}
    Consider the MADC setting with combinatorial topology, where there are $\Lambda$ mappers and $K = \binom{\Lambda}{\alpha}$ reducers for a fixed value of $\alpha \in [\Lambda]$. Then, the optimal communication load $L^{\star}(r)$ is upper bounded by $L_{\textnormal{UB}}(r)$ which is a piecewise linear curve with corner points
    \begin{equation}
        (r, L_{\textnormal{UB}}(r)) = \left(r, \frac{\binom{\Lambda - \alpha}{r}}{\binom{\Lambda}{r}\left( \binom{r + \alpha}{r} - 1 \right)} \right), \quad \forall r \in [\Lambda - \alpha + 1].
    \end{equation}
\end{theorem}
\begin{proof}
    The detailed proof of the scheme is reported in \Cref{sec: Achievability Without Download Cost Proof}, whereas an illustrative example is instead described in \Cref{sec: Illustrative Example}.
\end{proof}

As we already mentioned in \Cref{rem: Special Case}, if we set $\alpha = 1$ and we consider each mapper-reducer pair as a unique computing machine, we obtain the same system model in~\cite{Li2018FundamentalTradeoffComputation}. Interestingly, we can see that, if we specialize the result in \Cref{thm: Achievable Bound Without Download Cost} to the case $\alpha = 1$, we obtain the same computation-communication tradeoff as in~\cite[Theorem 1]{Li2018FundamentalTradeoffComputation}.

Another noteworthy aspect is the following. If we fix the number of mappers $\Lambda$ and the computation load $r$, then adding more reducers by increasing\footnote{Notice that the number $K = \binom{\Lambda}{\alpha}$ of reducers is actually increased as far as $\alpha \leq \Lambda/2$. The scenario where $\alpha > \Lambda/2$ is unrealistic and is not considered here.} the multi-access degree $\alpha$ will in fact entail a smaller communication load. This (perhaps surprising) outcome is most certainly not the result of each reducer requiring fewer intermediate values during the shuffle phase. Such decrease could not have compensated for the increasing $K$. Instead, this decrease in the communication load stems from the nature of the combinatorial topology, which allows each reducer to more efficiently use its side information to cancel interference in an accelerated manner. This is achieved because these reducers are connected to the mappers in a manner that effectively aligns the interference patterns. As one can imagine, if we increase the number of reducers and we properly connect each of them to multiple mappers, the achievable scheme in \Cref{thm: Achievable Bound Without Download Cost} outperforms the coded scheme in~\cite{Li2018FundamentalTradeoffComputation}. This is formally stated in the following corollary.

\begin{corollary}\label{cor: Multi-Access Degree Improvement}
    For fixed computation load $r$, the achievable communication load in \Cref{thm: Achievable Bound Without Download Cost} decreases for increasing $\alpha$, even though $K$ --- and so the corresponding speedup factor in computing reduce functions --- increases substantially.
\end{corollary}
\begin{proof}
    The proof is described in~\refappendix{cor: Multi-Access Degree Improvement Proof}.
\end{proof}

We proceed to construct an information-theoretic converse on the communication load of the MADC setting. As it will be pointed out in the general proof in \Cref{sec: Converse Without Download Cost Proof}, the construction of the converse takes inspiration from~\cite[Lemma 2]{Woolsey2021CombinatorialDesignCascaded} as well as from ideas in~\cite{Brunero2021FundamentalLimitsCombinatorial}. Essentially, the bound here manages to merge the approach in~\cite[Lemma 2]{Woolsey2021CombinatorialDesignCascaded}, where a converse bound is built using key properties of the entropy function, with the index coding techniques in~\cite{Brunero2021FundamentalLimitsCombinatorial}, where the nodes of a side information graph are iteratively selected in a proper way to systematically identify large acyclic subgraphs that are used to develop a tight converse. The result is formally stated in the following.

\begin{theorem}[Converse Bound]\label{thm: Converse Bound Without Download Cost}
    Consider the MADC setting with combinatorial topology, where there are $\Lambda$ mappers and $K = \binom{\Lambda}{\alpha}$ reducers for a fixed value of $\alpha \in [\Lambda]$. Then, the optimal communication load $L^{\star}(r)$ is lower bounded by $L_{\textnormal{LB}}(r)$ which is a piecewise linear curve with corner points
    \begin{equation}
        (r, L_{\textnormal{LB}}(r)) = \left(r, \frac{\binom{\Lambda}{r + \alpha}}{\binom{\Lambda}{\alpha}\binom{\Lambda}{r}} \right), \quad \forall r \in [\Lambda - \alpha + 1].
    \end{equation}
\end{theorem}
\begin{proof}
    The proof is described in \Cref{sec: Converse Without Download Cost Proof}.
\end{proof}

Finally, from the results in \Cref{thm: Achievable Bound Without Download Cost} and \Cref{thm: Converse Bound Without Download Cost}, we can provide an order optimality guarantee for the MADC model. Indeed, comparing the achievable performance and the converse bound, we conclude that the two are within a constant multiplicative gap. We see this in the following theorem\footnote{Notice that the order optimality result in \Cref{thm: Order Optimality Without Download Cost} excludes the value $\alpha = 1$. Indeed, it can be verified that for such case the achievable performance in \Cref{thm: Achievable Bound Without Download Cost} and the converse in \Cref{thm: Converse Bound Without Download Cost} are within a factor of at most $2$. However, we already know that the coded scheme in~\cite{Li2018FundamentalTradeoffComputation} is exactly optimal when $\alpha = 1$. Hence, such value is neglected when comparing the aforementioned results.}.

\begin{theorem}[Order Optimality]\label{thm: Order Optimality Without Download Cost}
    For the MADC system with combinatorial topology, $\Lambda$ mappers and $K = \binom{\Lambda}{\alpha}$ reducers for a fixed value of $\alpha \in [2:\Lambda]$, the achievable performance in \Cref{thm: Achievable Bound Without Download Cost} is within a factor of at most $1.5$ from the optimal.
\end{theorem}
\begin{proof}
    The proof is described in~\refappendix{sec: Order Optimality Without Download Cost Proof}.
\end{proof}

\begin{figure}[!htb]
  \centering
  \def\L{10}
  \def\a{2}
  \pgfmathsetmacro\rmax{\L - \a + 1}
  \begin{tikzpicture}[ 
    declare function = {
      binom(\n,\k) = (\n >= \k)*(\n!)/(\k!*(\n-\k)!);
      CDC(\K, \r) = (1 - \r/\K)/\r;
      UB(\L, \a, \r) = (binom(\L - \a, \r))/(binom(\L, \r)*(binom(\r + \a, \r) - 1));
      LB(\L, \a, \r) = (binom(\L, \r + \a)/(binom(\L, \r)*binom(\L, \a)));
    }]
    \begin{axis}[xlabel = {Computation Load $r$}, ylabel = {Communication Load $L(r)$}, grid = major, enlargelimits = {value = 0.2, upper}, legend cell align = {left}, legend style = {font = \small}, xtick distance = 1, scale = 1.25]
      \addplot+[variable = r, samples at = {1, ..., \L}, mark = *, mark size = 1.5, mark options = {solid}, thick, black, dashed]{CDC(\L, r)};
      \addlegendentry{CDC with $\Lambda = 10$ and $\alpha = 1$};
      \addplot+[variable = r, samples at = {1, ..., \rmax}, mark = diamond*, mark size = 2.5, mark options = {solid}, thick, blue, dash dot]{UB(\L, \a, r)};
      \addlegendentry{Achievable MADC with $\Lambda = 10$ and $\alpha = 2$};
      \addplot+[variable = r, samples at = {1, ..., \rmax}, mark = square*, mark size = 1.5, mark options = {solid}, thick, red, dash dot]{LB(\L, \a, r)};
      \addlegendentry{Converse MADC with $\Lambda = 10$ and $\alpha = 2$};
    \end{axis}
  \end{tikzpicture}
  \caption{Comparison between original CDC, where there are $\Lambda = 10$ pairs of mappers and reducers, and MADC with combinatorial topology, $\Lambda = 10$ mappers and $K = 45$ reducers, where each of them is uniquely associated to $\alpha = 2$ mappers.}
  \label{fig: Comparison Between CDC and MADC}
\end{figure}
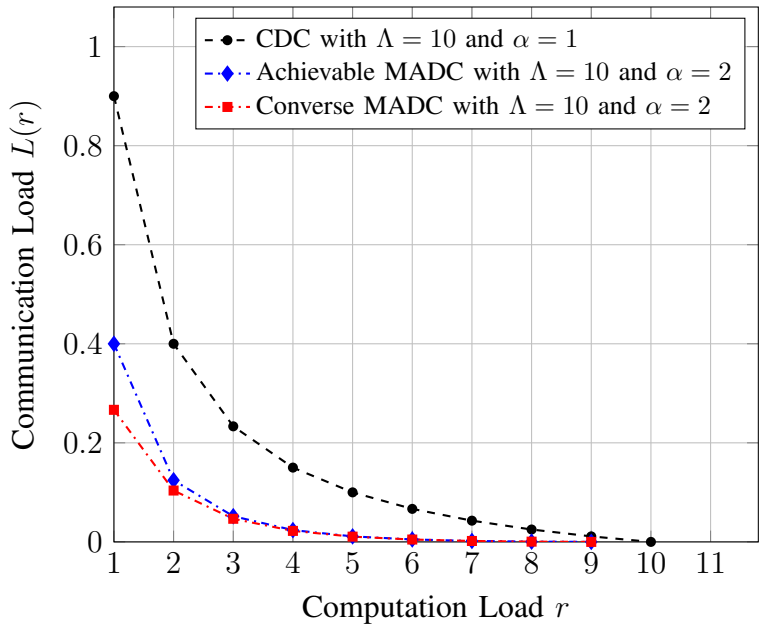

In \Cref{fig: Comparison Between CDC and MADC} we can see a comparison between the original CDC framework and the proposed MADC model. More specifically, for the first setting we consider $\Lambda = 10$ pairs of mappers and reducers, where each pair $\lambda \in [10]$ can be considered as a unique computing server having its own subset $\mathcal{M}_\lambda$ of assigned files. For the second setting we consider $\Lambda = 10$ mappers and $K = \binom{10}{2} = 45$ reducers, where there is a reducer connected to any $\alpha = 2$ mappers. According to \Cref{cor: Multi-Access Degree Improvement}, the achievable load in \Cref{thm: Achievable Bound Without Download Cost} decreases for increasing $\alpha$ and fixed computation load, as indicated by the diamond blue curve in \Cref{fig: Comparison Between CDC and MADC} which is well below the dot black counterpart corresponding to the original achievable scheme of coded distributed computing. Notice that the comparison in \Cref{fig: Comparison Between CDC and MADC} between the CDC setting with $\alpha = 1$ and the MADC setting with $\alpha > 1$ is fair for what concerns the computation-communication tradeoff: indeed, not only the computation load $r$ remains the same as far as the number of mappers $\Lambda$ stays the same, but also the number of reducers that need to communicate with each other is much larger than $\Lambda$ when $\alpha > 1$.

\begin{remark}
    We point out that comparing a setting where $\alpha = 1$ with a setting where $\alpha > 1$ offers noteworthy insights. Indeed, even though one could expect the communication load to be reduced when $\alpha > 1$ --- as in such case each reducer accesses more than one mapper and consequently misses less intermediate values --- it is also true that the number of reducers itself increases as $\alpha$ increases. Consequently, there is an undeniable tension between the higher multi-access degree $\alpha > 1$ for each reducer, which implies less data needed by each reducer, but also implies a larger number of reducers in the system. For these reasons, it is interesting to notice how the network topology plays a fundamental role in resolving such tension by appropriately shaping the interference patterns. As a consequence, the communication load ultimately decreases as the number of reducers $K = \binom{\Lambda}{\alpha}$ increases as far as each of them is \emph{properly} connected to $\alpha$ mappers.
\end{remark}

\begin{figure}[!htb]
  \centering
  \begin{tikzpicture}[ 
    declare function = {
      binom(\n,\k) = (\n >= \k)*(\n!)/(\k!*(\n-\k)!);
      G(\a, \r) = (binom(\r + \a, \r) - 1);
    }]
    \begin{axis}[xlabel = {Computation Load $r$}, ylabel = {Coding Gain}, grid = major, enlargelimits = {value = 0.2, upper}, legend cell align = {left}, legend style = {font = \small}, legend pos = {north west}, xtick distance = 1, scale = 1.25]
      \addplot+[variable = r, samples at = {1, ..., 6}, mark = *, mark size = 1.5, mark options = {solid}, thick, black, dashed]{G(1, r)};
      \addlegendentry{Gain with $\alpha = 1$};
      \addplot+[variable = r, samples at = {1, ..., 6}, mark = diamond*, mark size = 2.5, mark options = {solid}, thick, teal, dashed]{G(2, r)};
      \addlegendentry{Gain with $\alpha = 2$};
      \addplot+[variable = r, samples at = {1, ..., 6}, mark = square*, mark size = 1.5, mark options = {solid}, thick, purple, dashed]{G(3, r)};
      \addlegendentry{Gain with $\alpha = 3$};
    \end{axis}
  \end{tikzpicture}
  \caption{Comparison between the coding gain for different values of $\alpha$ as a function of the computation load $r$. We recall that $\alpha = 1$ corresponds to the original CDC framework.}
  \label{fig: Comparison Between CDC and MADC Gains}
\end{figure}
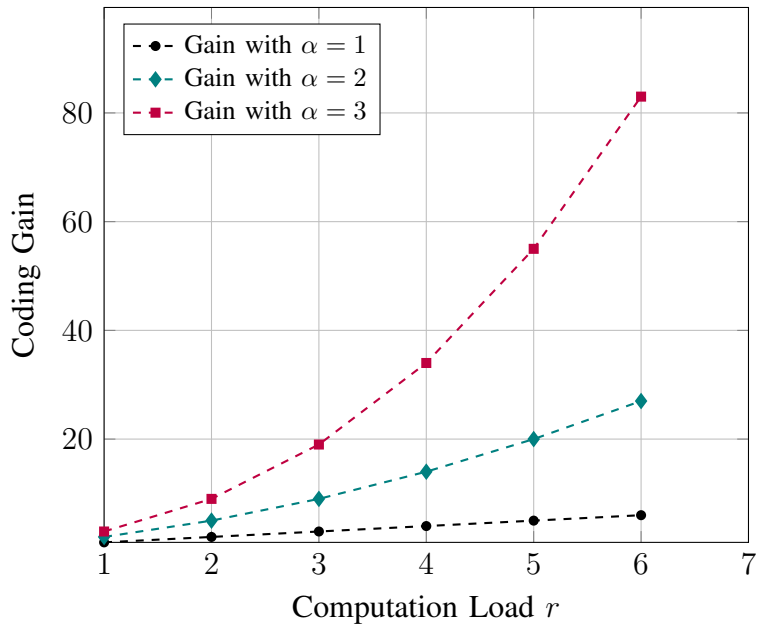

A further comparison is provided in \Cref{fig: Comparison Between CDC and MADC Gains} which focuses on the coding gains. As we can see, the gains brought about by the multi-access setting are impressive even when the computation load is small, which is the regime of interest in practical settings.

\subsection{Characterizing the Max-Link Load}\label{sec: Characterizing the Max-Link Load}

We now consider a distributed computing scenario where the download cost is not negligible. The following describes the achievable max-link communication load that captures both communication and download costs. 

\begin{theorem}[Achievable Bound]\label{thm: Achievable Bound With Download Cost}
    Consider the MADC setting with combinatorial topology, where there are $\Lambda$ mappers and $K = \binom{\Lambda}{\alpha}$ reducers for a fixed value of $\alpha \in [\Lambda]$. Then, the optimal max-link communication load $L^{\star}_{\textnormal{max-link}}(r)$ is upper bounded by $L_{\textnormal{max-link}, \textnormal{UB}}(r)$ which is given by
    \begin{equation}
        L_{\textnormal{max-link}, \textnormal{UB}}(r) = \max\left(\sum_{j \in [\Lambda]} \frac{\binom{\Lambda - \alpha}{j}}{\binom{\Lambda}{j}\left(\binom{j + \alpha}{j} - 1\right)} \frac{\tilde{a}^{j}_{\star}}{N}, \sum_{j \in [\Lambda]} \frac{\binom{\Lambda}{\alpha} - \binom{\Lambda - j}{\alpha}}{\alpha \binom{\Lambda}{\alpha}} \frac{\tilde{a}^{j}_{\star}}{N}  \right)
    \end{equation}
    where the vector $\tilde{\bm{a}}_{\star} = (\tilde{a}^{1}_{\star}, \dots, \tilde{a}^{\Lambda}_{\star})$ is the optimal solution to the linear program
    \begin{subequations}\label{eqn: Linear Program}
        \begin{alignat}{2}
                & \min_{\tilde{\bm{a}}_{\mathcal{M}}} & \quad & \frac{1}{2} \sum_{j \in [\Lambda]} \left(\frac{\binom{\Lambda}{\alpha + j}}{\binom{\Lambda}{\alpha}\binom{\Lambda}{j}} + \frac{\binom{\Lambda}{\alpha} - \binom{\Lambda - j}{\alpha}}{\alpha \binom{\Lambda}{\alpha}} \right) \frac{\tilde{a}^{j}_{\mathcal{M}}}{N} \\
                & \textnormal{subject to}  &  & \tilde{a}^{j}_{\mathcal{M}} \geq 0, \quad \forall j \in [\Lambda] \\
                & & & \sum_{j \in [\Lambda]} \frac{\tilde{a}^{j}_{\mathcal{M}}}{N} = 1 \\
                & & & \sum_{j \in [\Lambda]} j\frac{\tilde{a}^{j}_{\mathcal{M}}}{N} \leq r \label{eqn: Computation Load Constraint}
        \end{alignat}
    \end{subequations}
    and where $\tilde{\bm{a}}_{\mathcal{M}} = (\tilde{a}^{1}_{\mathcal{M}}, \dots, \tilde{a}^{\Lambda}_{\mathcal{M}})$ is the control variable.
\end{theorem}
\begin{proof}
    The detailed proof of the scheme is reported in \Cref{sec: Achievability With Download Cost Proof}.
\end{proof}

We proceed by proposing an information-theoretic converse on the max-link communication load. The result is presented in the following theorem.

\begin{theorem}[Converse Bound]\label{thm: Converse Bound With Download Cost}
    Consider the MADC setting with combinatorial topology, where there are $\Lambda$ mappers and $K = \binom{\Lambda}{\alpha}$ reducers for a fixed value of $\alpha \in [\Lambda]$. Then, the optimal max-link communication load $L^{\star}_{\textnormal{max-link}}(r)$ is lower bounded by $L_{\textnormal{max-link}, \textnormal{LB}}(r)$ which is given by
    \begin{equation}
        L_{\textnormal{max-link}, \textnormal{LB}}(r) = \frac{1}{2} \sum_{j \in [\Lambda]} \left(\frac{\binom{\Lambda}{\alpha + j}}{\binom{\Lambda}{\alpha}\binom{\Lambda}{j}} + \frac{\binom{\Lambda}{\alpha} - \binom{\Lambda - j}{\alpha}}{\alpha \binom{\Lambda}{\alpha}} \right) \frac{\tilde{a}^{j}_{\star}}{N}
    \end{equation}
    where the vector $\tilde{\bm{a}}_{\star} = (\tilde{a}^{1}_{\star}, \dots, \tilde{a}^{\Lambda}_{\star})$ is the optimal solution to the linear program in~\eqref{eqn: Linear Program}.
\end{theorem}
\begin{proof}
    The proof is described in \Cref{sec: Converse With Download Cost Proof}.
\end{proof}

Finally, we can compare the results in \Cref{thm: Achievable Bound With Download Cost} and \Cref{thm: Converse Bound With Download Cost} to establish the gap to optimality of the achievable performance in \Cref{thm: Achievable Bound With Download Cost}. Notice that now we do not exclude the value $\alpha = 1$ for such comparison, since for such case there is no previously known optimality result to the best of our knowledge.

\begin{theorem}[Order Optimality]\label{thm: Order Optimality With Download Cost}
    For the MADC system with combinatorial topology, $\Lambda$ mappers and $K = \binom{\Lambda}{\alpha}$ reducers for a fixed value of $\alpha \in [\Lambda]$, the achievable performance in \Cref{thm: Achievable Bound With Download Cost} is within a factor of at most $4$ from the optimal.
\end{theorem}
\begin{proof}
    The proof is described in~\refappendix{sec: Order Optimality With Download Cost Proof}.
\end{proof}

\begin{remark}
    Interestingly, the results in this section can also be derived --- while keeping the same constant factor of at most $4$ from the optimal --- for a weighted max-link load defined as $L_{\text{max-link}, \beta} \coloneqq \max(L, \beta J)$ for some $\beta \geq 0$. Notice that such metric is of particular interest whenever we want to account for a \emph{weighted} download cost $\beta J$, in which case the cost of communicating data from mappers to reducers is \emph{proportional} to the maximal number of normalized bits flowing from any one mapper to any one reducer. The extreme case $\beta = 0$ corresponds to the communication load $L_{\text{max-link}, 0} = L$ characterized in \Cref{sec:  Characterizing the Communication Load}, whereas the case $\beta = 1$ implies the max-link load $L_{\text{max-link}, 1} = L_{\text{max-link}}$ investigated in the current section.
\end{remark}

\section{Illustrative Example of the Coded Scheme}\label{sec: Illustrative Example}

In this section, we propose an illustrative example of the coded scheme which will be later presented in its general form in \Cref{sec: Achievability Without Download Cost Proof}. The example refers to the schematic in \Cref{fig: Multi-Access Distributed Computing Example}. Notice that this section aims to provide an example for the achievable scheme in \Cref{thm: Achievable Bound Without Download Cost}, consequently the download cost will be neglected.

We consider $\Lambda = 4$ mappers and $K = \binom{\Lambda}{\alpha} = 6$ reducers where $\alpha = 2$. We assume to have $Q = 12$ output functions to be computed across the ensemble of mappers and reducers, $N = 8$ input files $\{w_n : n \in [8]\}$ and computation load $r = 1$. Recalling that the computation load is defined as the normalized number of files which are mapped across the $\Lambda$ map nodes, having computation load equal to $r = 1$ implies that the number of files mapped across all mappers is equal to the size of the input data set, i.e., $N$ files. Under the uniform function assignment assumption, we assign $Q/K = 2$ output functions to each reducer\footnote{For the sake of simplicity, we will often omit braces and commas when indicating sets, e.g., the reducer $\{1, 2\}$, which is connected to mappers $1$ and $2$, can be simply denoted as $12$.} $\mathcal{U} \in [4]_2 = \{12, 13, 14, 23, 24, 34\}$. Thus, recalling that $\mathcal{W}_{\mathcal{U}}$ represents the indices of the reduce functions assigned to reducer $\mathcal{U} \in [4]_2$, we arbitrarily let
\begin{align}
    \mathcal{W}_{12} & = \{1, 2\} \\
    \mathcal{W}_{13} & = \{3, 4\} \\
    \mathcal{W}_{14} & = \{5, 6\} \\
    \mathcal{W}_{23} & = \{7, 8\} \\
    \mathcal{W}_{24} & = \{9, 10\} \\
    \mathcal{W}_{34} & = \{11, 12\}.
\end{align}

\subsection{Map Phase}

This phase requires the input files to be split among mappers and so we proceed by grouping the $N = 8$ files into $4$ batches $\mathcal{B}_\lambda$ for each $\lambda \in [4]$ as
\begin{align}
    \mathcal{B}_1 & = \{w_1, w_2\} \\
    \mathcal{B}_2 & = \{w_3, w_4\} \\
    \mathcal{B}_3 & = \{w_5, w_6\} \\
    \mathcal{B}_4 & = \{w_7, w_8\}.
\end{align}
Then, recalling that $\mathcal{M}_\lambda$ represents the set of files mapped by the mapper $\lambda \in [\Lambda]$, we set here $\mathcal{M}_\lambda = \{\mathcal{B}_{\mathcal{T}_1} : \mathcal{T}_1 \subseteq [4], |\mathcal{T}_1| = 1, \lambda \in \mathcal{T}_1\} = \{\mathcal{B}_\lambda\}$.

Since each mapper is assigned $2$ input files, we have that $|\mathcal{M}_\lambda| = |\mathcal{B}_\lambda| = 2$ for each $\lambda \in [4]$. Hence, we can check that such file assignment satisfies the computation load constraint $r = 1$, as indeed we have
\begin{equation}
    \frac{\sum_{\lambda \in [4]} |\mathcal{M}_\lambda|}{N} = \frac{|\mathcal{M}_1| + |\mathcal{M}_2| + |\mathcal{M}_3| + |\mathcal{M}_4| }{8} = 1.
\end{equation}

Each mapper computes $Q = 12$ intermediate values for each assigned input file. In particular, recalling that $\mathcal{V}_{\lambda}$ is the set of IVs computed by the mapper $\lambda \in [4]$, we have
\begin{align}
    \mathcal{V}_1 & = \{v_{q, n}: q \in [12], w_n \in \mathcal{M}_1\} \\
    \mathcal{V}_2 & = \{v_{q, n}: q \in [12], w_n \in \mathcal{M}_2\} \\
    \mathcal{V}_3 & = \{v_{q, n}: q \in [12], w_n \in \mathcal{M}_3\} \\
    \mathcal{V}_4 & = \{v_{q, n}: q \in [12], w_n \in \mathcal{M}_4\}.
\end{align}

For example, since mapper $1$ is assigned the files in $\mathcal{M}_1 = \{w_1, w_2\}$, it will compute the intermediate values $v_{q, 1}$ and $v_{q, 2}$ for all $q \in [12]$. Then, mapper $2$ will compute the intermediate values $v_{q, 3}$ and $v_{q, 4}$ for all $q \in [12]$ since $\mathcal{M}_2 = \{w_3, w_4\}$. Similarly, mapper $3$ and mapper $4$ will compute the intermediate values $v_{q, 5}$ and $v_{q, 6}$, and $v_{q, 7}$ and $v_{q, 8}$, respectively, for all $q \in [12]$ since $\mathcal{M}_{3} = \{w_{5}, w_{6}\}$ and $\mathcal{M}_{4} = \{w_{7}, w_{8}\}$.

Now, considering that there is a reducer connected to any $\alpha = 2$ mappers, we know that each reducer $\mathcal{U} \in [4]_2$ can retrieve the IVs computed by the $2$ mappers in $\mathcal{U}$. Recalling that $\mathcal{V}_{\mathcal{U}}$ denotes the union set of IVs computed by the mappers in $\mathcal{U}$, we have
\begin{align}
    \mathcal{V}_{12} & = \{v_{q, n}: q \in [12], w_n \in \mathcal{M}_{12}\} \\
    \mathcal{V}_{13} & = \{v_{q, n}: q \in [12], w_n \in \mathcal{M}_{13}\} \\
    \mathcal{V}_{14} & = \{v_{q, n}: q \in [12], w_n \in \mathcal{M}_{14}\} \\
    \mathcal{V}_{23} & = \{v_{q, n}: q \in [12], w_n \in \mathcal{M}_{23}\} \\
    \mathcal{V}_{24} & = \{v_{q, n}: q \in [12], w_n \in \mathcal{M}_{24}\} \\
    \mathcal{V}_{34} & = \{v_{q, n}: q \in [12], w_n \in \mathcal{M}_{34}\}.
\end{align}

\subsection{Shuffle Phase}

We describe now how each reducer $\mathcal{U} \in [4]_2$ constructs its multicast message $X_{\mathcal{U}}$. Since the procedure is the same for each reducer, we continue our example by focusing for simplicity on reducer $\{1, 2\}$ only.

First of all, we let $\mathcal{S} \subseteq ([4] \setminus \{1, 2\})$ with $|\mathcal{S}| = 1$. Then, for each $\mathcal{R} \subseteq (\mathcal{S} \cup \{1, 2\})$ such that $|\mathcal{R}| = 2$ and $\mathcal{R} \neq \{1, 2\}$, and for $\mathcal{T}_1 = (\mathcal{S} \cup \{1, 2\}) \setminus \mathcal{R}$, reducer $\{1, 2\}$ concatenates the intermediate values $\{v_{q, n} : q \in \mathcal{W}_\mathcal{R}, w_n \in \mathcal{B}_{\mathcal{T}_1}\}$ into the symbol $U_{\mathcal{W}_\mathcal{R}, \mathcal{T}_1} = (v_{q, n} : q \in \mathcal{W}_\mathcal{R}, w_n \in \mathcal{B}_{\mathcal{T}_1})$. Notice that having $\mathcal{R} \neq \{1, 2\}$ implies that $\mathcal{T}_1 \cap \{1, 2\} \neq \emptyset$, so reducer $\{1, 2\}$ can retrieve $\mathcal{B}_{\mathcal{T}_1}$ from the mappers it is connected to and can construct the symbol $U_{\mathcal{W}_\mathcal{R}, \mathcal{T}_1}$. Subsequently, such symbol is evenly split as
\begin{equation}
    U_{\mathcal{\mathcal{W}_\mathcal{R}}, \mathcal{T}_1} = \left(U_{\mathcal{\mathcal{W}_\mathcal{R}}, \mathcal{T}_1, \mathcal{T}_2} : \mathcal{T}_2 \subseteq (\mathcal{R} \cup \mathcal{T}_1), |\mathcal{T}_2| = 2, \mathcal{T}_2 \neq \mathcal{R} \right).
\end{equation}
This means that when, say, $\mathcal{S} = \{3\}$, reducer $\{1, 2\}$ creates the symbols
\begin{align}
    U_{\mathcal{W}_{13}, 2} & = (v_{q, n} : q \in \mathcal{W}_{13}, w_n \in \mathcal{B}_2) \\
    U_{\mathcal{W}_{23}, 1} & = (v_{q, n} : q \in \mathcal{W}_{23}, w_n \in  \mathcal{B}_1)
\end{align}
and when $\mathcal{S} = \{4\}$, the same reducer $\{1, 2\}$ creates the symbols
\begin{align}
    U_{\mathcal{W}_{14}, 2} & = (v_{q, n} : q \in \mathcal{W}_{14}, w_n \in \mathcal{B}_2) \\
    U_{\mathcal{W}_{24}, 1} & = (v_{q, n} : q \in \mathcal{W}_{24}, w_n \in \mathcal{B}_1).
\end{align}
Each of the symbols above is then evenly split in two segments as
\begin{align}
    U_{\mathcal{W}_{13}, 2} & = (U_{\mathcal{W}_{13}, 2, 12}, U_{\mathcal{W}_{13}, 2, 23}) \\
    U_{\mathcal{W}_{23}, 1} & = (U_{\mathcal{W}_{23}, 1, 12}, U_{\mathcal{W}_{23}, 1, 13}) \\
    U_{\mathcal{W}_{14}, 2} & = (U_{\mathcal{W}_{14}, 2, 12}, U_{\mathcal{W}_{14}, 2, 24}) \\
    U_{\mathcal{W}_{24}, 1} & = (U_{\mathcal{W}_{24}, 1, 12}, U_{\mathcal{W}_{24}, 1, 14}).
\end{align}
Next, reducer $\{1, 2\}$ constructs the coded message
\begin{equation}
    \bigoplus_{\mathcal{R} \subseteq (\mathcal{S} \cup \{1, 2\}): |\mathcal{R}| = 2, \mathcal{R} \neq \{1, 2\}} U_{\mathcal{W}_{\mathcal{R}}, (\mathcal{S} \cup \{1, 2\}) \setminus \mathcal{R}, 12}
\end{equation}
for each $\mathcal{S} \subseteq ([4] \setminus \{1, 2\})$ with $|\mathcal{S}| = 1$, and concatenates all of them to form $X_{12}$, which is given by
\begin{align}
    X_{12} & =\left( \bigoplus_{\mathcal{R} \subseteq (\mathcal{S} \cup \{1, 2\}): |\mathcal{R}| = 2, \mathcal{R} \neq \{1, 2\}} U_{\mathcal{W}_{\mathcal{R}}, (\mathcal{S} \cup \{1, 2\}) \setminus \mathcal{R}, 12} : \mathcal{S} \subseteq ([\Lambda] \setminus \{1, 2\}), |\mathcal{S}| = 1 \right) \\ 
    & = (U_{\mathcal{W}_{13}, 2, 12} \oplus U_{\mathcal{W}_{23}, 1, 12}, U_{\mathcal{W}_{24}, 1, 12} \oplus U_{\mathcal{W}_{14}, 2, 12}).
\end{align}

Similarly, each other reducer prepares and multicasts its message following the procedure described above. In the end, the following messages
\begin{align}
    X_{12} & = (U_{\mathcal{W}_{13}, 2, 12} \oplus U_{\mathcal{W}_{23}, 1, 12}, U_{\mathcal{W}_{24}, 1, 12} \oplus U_{\mathcal{W}_{14}, 2, 12}) \\
    X_{13} & = (U_{\mathcal{W}_{12}, 3, 13} \oplus U_{\mathcal{W}_{23}, 1, 13}, U_{\mathcal{W}_{14}, 3, 13} \oplus U_{\mathcal{W}_{34}, 1, 13}) \\
    X_{14} & = (U_{\mathcal{W}_{12}, 4, 14} \oplus U_{\mathcal{W}_{24}, 1, 14}, U_{\mathcal{W}_{13}, 4, 14} \oplus U_{\mathcal{W}_{34}, 1, 14}) \\
    X_{23} & = (U_{\mathcal{W}_{12}, 3, 23} \oplus U_{\mathcal{W}_{13}, 2, 23}, U_{\mathcal{W}_{24}, 3, 23} \oplus U_{\mathcal{W}_{34}, 2, 23}) \\
    X_{24} & = (U_{\mathcal{W}_{12}, 4, 24} \oplus U_{\mathcal{W}_{14}, 2, 24}, U_{\mathcal{W}_{23}, 4, 24} \oplus U_{\mathcal{W}_{34}, 2, 24}) \\
    X_{34} & = (U_{\mathcal{W}_{13}, 4, 34} \oplus U_{\mathcal{W}_{14}, 3, 34}, U_{\mathcal{W}_{23}, 4, 34} \oplus U_{\mathcal{W}_{24}, 3, 34})
\end{align}
are exchanged among the reducers on the common-bus link during the shuffle phase.

\subsection{Reduce Phase}

As when describing the shuffle phase, we can again focus on reducer $\{1, 2\}$ and observe how it correctly computes the reduce functions in $\mathcal{W}_{12}$ by using the set of multicast messages $\{X_{\mathcal{U}} : \mathcal{U} \in [4]_2\}$ and the set $\mathcal{V}_{12}$ of IVs which the reducer $\{1, 2\}$ can access. Indeed, a similar procedure can be shown for all other reducers.

First of all, reducer $\{1, 2\}$ needs the IVs $\{v_{q, n} : q \in \mathcal{W}_{12}, n \in [8]\}$ to compute the reduce functions in $\mathcal{W}_{12}$. Since such reducer has already access to the IVs in $\mathcal{V}_{12}$, it can obtain the symbols $U_{\mathcal{W}_{12}, 1}$ and $U_{\mathcal{W}_{12}, 2}$. However, it misses the intermediate values $\{v_{q, n} : q \in \mathcal{W}_{12}, w_n \notin \mathcal{M}_{12}\}$ or, similarly, it misses the symbols $U_{\mathcal{W}_{12}, 3} = (v_{q, n} : q \in \mathcal{W}_{12}, w_n \in \mathcal{B}_3)$ and $U_{\mathcal{W}_{12}, 4} = (v_{q, n} : q \in \mathcal{W}_{12}, w_n \in \mathcal{B}_4)$. We see now how these symbols can be obtained from the set of multicast messages.

During the shuffle procedure, each symbol is split in two even segments, so, consequently, symbols $U_{\mathcal{W}_{12}, 3}$ and $U_{\mathcal{W}_{12}, 4}$ are split as
\begin{align}
    U_{\mathcal{W}_{12}, 3} & = (U_{\mathcal{W}_{12}, 3, 13}, U_{\mathcal{W}_{12}, 3, 23})\\
    U_{\mathcal{W}_{12}, 4} & = (U_{\mathcal{W}_{12}, 4, 14}, U_{\mathcal{W}_{12}, 4, 24}).
\end{align}
Now, reducer $\{1, 2\}$ can decode $U_{\mathcal{W}_{12}, 3, 13}$ from the message $X_{13}$. Indeed, the term $U_{\mathcal{W}_{12}, 3, 13} \oplus U_{\mathcal{W}_{23}, 1, 13}$ appears in $X_{13}$ and reducer $\{1, 2\}$ can use the IVs in $\mathcal{V}_{12}$ to cancel the interference term $U_{\mathcal{W}_{23}, 1, 13}$. Similarly, the term $U_{\mathcal{W}_{12}, 3, 23} \oplus U_{\mathcal{W}_{13}, 2, 23}$ appears in $X_{23}$, where again the interference $U_{\mathcal{W}_{13}, 2, 23}$ can be canceled by means of the IVs retrieved by the mappers $1$ and $2$. Hence, reducer $\{1, 2\}$ successfully decodes $U_{\mathcal{W}_{12}, 3, 13}$ and $U_{\mathcal{W}_{12}, 3, 23}$ from the multicasted messages $X_{13}$ and $X_{23}$, reconstructing then the symbol $U_{\mathcal{W}_{12}, 3}  = (U_{\mathcal{W}_{12}, 3, 13}, U_{\mathcal{W}_{12}, 3, 23})$. A similar procedure holds for reducer $\{1, 2\}$ to successfully reconstruct the symbol $U_{\mathcal{W}_{12}, 4}$, whose two segments are decoded from messages $X_{14}$ and $X_{24}$. Further, a similar procedure holds for any other reducer. Thus, we can conclude that every reducer is able to compute the assigned reduce functions after recovering the missing intermediate values from the messages multicasted by all reducers during the shuffle phase.

\subsection{Communication Load}

Recalling that the communication load is defined as the total number of bits transmitted by the $K$ reducers during the shuffle phase (normalized by the number of bits of all intermediate values), we wish to compute for this example this load, which takes the form
\begin{equation}
    L_{\textnormal{UB}}(r = 1) = \frac{\sum_{\mathcal{U} \in [\Lambda]_{\alpha}} |X_{\mathcal{U}}|}{QNT} = \frac{\sum_{\mathcal{U} \in [4]_{2}} |X_{\mathcal{U}}|}{96T}. 
\end{equation}
Since $|X_{\mathcal{U}}|$ is the same for each $\mathcal{U} \in [4]_{2}$, we focus again on reducer $\{1, 2\}$ and its multicast transmission $X_{12} = (U_{\mathcal{W}_{13}, 2, 12} \oplus U_{\mathcal{W}_{23}, 1, 12}, U_{\mathcal{W}_{24}, 1, 12} \oplus U_{\mathcal{W}_{14}, 2, 12})$, which contains two XOR messages. Focusing on the first message $U_{\mathcal{W}_{13}, 2, 12} \oplus U_{\mathcal{W}_{23}, 1, 12}$, we can see that it is a XOR composed of two segments, i.e., one segment for the symbol $U_{\mathcal{W}_{13}, 2}$ and one segment for the symbol $U_{\mathcal{W}_{23}, 1}$. Since the size of each symbol is $4T$ bits, the resulting XOR message has size $2T$ bits. Hence, given that $X_{12}$ contains two XOR messages, we can conclude that $|X_{12}| = 4T$ bits. Consequently, the resulting achievable communication load is given by
\begin{align}
    L_{\textnormal{UB}}(r = 1) = \frac{\sum_{\mathcal{U} \in [4]_{2}} |X_{\mathcal{U}}|}{96T} = \frac{24T}{96T} = \frac{1}{4}. 
\end{align}
Using the converse in \Cref{thm: Converse Bound Without Download Cost}, we can see that the achievable performance above is within a factor $1.5$ from the optimal.

\section{Proof of Achievable Bound in \texorpdfstring{\Cref{thm: Achievable Bound Without Download Cost}}{Theorem~\ref{thm: Achievable Bound Without Download Cost}}}\label{sec: Achievability Without Download Cost Proof}

We assume that there are $\Lambda$ mappers and $K = \binom{\Lambda}{\alpha}$ reducers, and we assume the aforementioned combinatorial topology where each reducer is exactly and uniquely connected to $\alpha$ mappers. We then consider some arbitrary computation load $r \in [\Lambda - \alpha + 1]$ and we consider $Q = \eta_2 K$ output functions with $\eta_2 \in \mathbb{N}^{+}$, allowing us to separate the $Q$ functions into $K$ disjoint groups $\mathcal{W}_{\mathcal{U}}$ for each $\mathcal{U} \in [\Lambda]_\alpha$, so that each reducer is assigned $\eta_2$ functions, corresponding to $|\mathcal{W}_{\mathcal{U}}| = \eta_2$ for each $\mathcal{U} \in [\Lambda]_{\alpha}$.

\subsection{Map Phase}

First, the input database is split in $\binom{\Lambda}{r}$ disjoint batches, each containing $\eta_1 = N/\binom{\Lambda}{r}$ files, where we assume that $N$ is large enough such that $\eta_1 \in \mathbb{N}^{+}$. Consequently, we have a batch of files for each $\mathcal{T}_1 \subseteq [\Lambda]$ such that $|\mathcal{T}_1| = r$, which implies
\begin{equation}
    \{w_1, \dots, w_N\} = \bigcup_{\mathcal{T}_1 \subseteq [\Lambda] : |\mathcal{T}_1| = r} \mathcal{B}_{\mathcal{T}_1}
\end{equation}
where we denote by $\mathcal{B}_{\mathcal{T}_1}$ the batch of $\eta_1$ files associated with the label $\mathcal{T}_1$. Then, mapper $\lambda \in [\Lambda]$ is assigned all batches $\mathcal{B}_{\mathcal{T}_1}$ having $\lambda \in \mathcal{T}_1$, which means that
\begin{equation}
    \mathcal{M}_\lambda = \{\mathcal{B}_{\mathcal{T}_1} : \mathcal{T}_1 \subseteq [\Lambda], |\mathcal{T}_1| = r, \lambda \in \mathcal{T}_1\}.
\end{equation}
We can see that the computation load constraint is satisfied, since we have
\begin{equation}
    \frac{\sum_{\lambda \in [\Lambda]} |\mathcal{M}_\lambda|}{N}  = \frac{\Lambda \eta_1 \binom{\Lambda - 1}{r - 1}}{\eta_1 \binom{\Lambda}{r}} = r
\end{equation}

Then, each mapper computes $Q$ intermediate values for each assigned input file, so for each $\lambda \in [\Lambda]$ we have $\mathcal{V}_{\lambda} = \{v_{q, n} : q \in [Q], w_n \in \mathcal{M}_{\lambda}\}$. Since then each reducer has access to $\alpha$ mappers, reducer $\mathcal{U} \in [\Lambda]_\alpha$ can retrieve\footnote{Since we are presenting here the proof of the achievable bound in \Cref{thm: Achievable Bound Without Download Cost}, we will neglect the download cost, assuming consequently that each reducer can access the IVs without any additional communication cost.} the intermediate values in $\mathcal{V}_{\mathcal{U}} = \{v_{q, n} : q \in [Q], w_n \in \mathcal{M}_{\mathcal{U}}\}$ recalling that $\mathcal{M}_{\mathcal{U}} = \cup_{\lambda \in \mathcal{U}} \mathcal{M}_\lambda$. Since $|\mathcal{V}_{\mathcal{U}}| = Q\eta_1\left( \binom{\Lambda}{r} - \binom{\Lambda - \alpha}{r}\right)$ for each $\mathcal{U} \in [\Lambda]_{\alpha}$, we can conclude that each computing node has access to all the intermediate values when $r \geq \Lambda - \alpha + 1$. Hence, we focus on the non-trivial regime $r \in [\Lambda - \alpha + 1]$ for any given $\Lambda$ and $\alpha$.

\subsection{Shuffle Phase}

Consider reducer $\mathcal{U} \in [\Lambda]_\alpha$. Let $\mathcal{S} \subseteq ([\Lambda] \setminus \mathcal{U})$ with $|\mathcal{S}| = r$. First, for each $\mathcal{R} \subseteq (\mathcal{S} \cup \mathcal{U})$ such that $|\mathcal{R}| = \alpha$ and $\mathcal{R} \neq \mathcal{U}$, and for $\mathcal{T}_1 = (\mathcal{S} \cup \mathcal{U}) \setminus \mathcal{R}$, reducer $\mathcal{U}$ concatenates the intermediate values $\{v_{q, n} : q \in \mathcal{W}_{\mathcal{R}}, w_n \in \mathcal{B}_{\mathcal{T}_1}\}$ into the symbol $U_{\mathcal{W}_{\mathcal{R}}, \mathcal{T}_1} = \left( v_{q, n} : q \in \mathcal{W}_{\mathcal{R}}, w_n \in \mathcal{B}_{\mathcal{T}_1} \right) \in \mathbb{F}_{2^{\eta_2 \eta_1 T}}$. Subsequently, such symbol is evenly split in $\left( \binom{r + \alpha}{r} - 1 \right)$ segments as
\begin{equation}
    U_{\mathcal{\mathcal{W}_\mathcal{R}}, \mathcal{T}_1} = \left(U_{\mathcal{\mathcal{W}_\mathcal{R}}, \mathcal{T}_1, \mathcal{T}_2} : \mathcal{T}_2 \subseteq (\mathcal{R} \cup \mathcal{T}_1), |\mathcal{T}_2| = \alpha, \mathcal{T}_2 \neq \mathcal{R} \right).
\end{equation}
Then, reducer $\mathcal{U}$ constructs the coded message
\begin{equation}
    \bigoplus_{\mathcal{R} \subseteq (\mathcal{S} \cup \mathcal{U}): |\mathcal{R}| = \alpha, \mathcal{R} \neq \mathcal{U}} U_{\mathcal{W}_{\mathcal{R}}, (\mathcal{S} \cup \mathcal{U}) \setminus \mathcal{R}, \mathcal{U}}
\end{equation}
for each $\mathcal{S} \subseteq ([\Lambda] \setminus \mathcal{U})$ with $|\mathcal{S}| = r$, and finally concatenates all of them into the following message
\begin{equation}
    X_{\mathcal{U}} = \left( \bigoplus_{\mathcal{R} \subseteq (\mathcal{S} \cup \mathcal{U}): |\mathcal{R}| = \alpha, \mathcal{R} \neq \mathcal{U}} U_{\mathcal{W}_{\mathcal{R}}, (\mathcal{S} \cup \mathcal{U}) \setminus \mathcal{R}, \mathcal{U}} : \mathcal{S} \subseteq ([\Lambda] \setminus \mathcal{U}), |\mathcal{S}| = r \right)
\end{equation}
which is multicasted to all other reducers via the error-free broadcast channel.

\subsection{Reduce Phase}

Consider reducer $\mathcal{U} \in [\Lambda]_\alpha$. Since such reducer is connected to $\alpha$ mappers, it misses a total of $\eta_2 \eta_1\binom{\Lambda - \alpha}{r}$ intermediate values, i.e, it misses $\eta_2$ intermediate values for each of the $\eta_1$ files in each batch that is not assigned to the mappers in $\mathcal{U}$. More precisely, reducer $\mathcal{U}$ misses the symbol $U_{\mathcal{W}_{\mathcal{U}}, \mathcal{T}_1}$ for each $\mathcal{T}_1 \subseteq ([\Lambda] \setminus \mathcal{U})$ with $|\mathcal{T}_1| = r$. We know that during the shuffle phase such symbol is evenly split in $\left( \binom{r + \alpha}{r} - 1 \right)$ segments as
\begin{equation}
    U_{\mathcal{W}_{\mathcal{U}}, \mathcal{T}_1} = \left(U_{\mathcal{W}_{\mathcal{U}}, \mathcal{T}_1, \mathcal{T}_2} : \mathcal{T}_2 \subseteq (\mathcal{U} \cup \mathcal{T}_1), |\mathcal{T}_2| = \alpha, \mathcal{T}_2 \neq \mathcal{U} \right).
\end{equation}
For each $\mathcal{T}_2 \subseteq (\mathcal{U} \cup \mathcal{T}_1)$ with $|\mathcal{T}_2| = \alpha$ and $\mathcal{T}_2 \neq \mathcal{U}$, we can verify that reducer $\mathcal{U}$ can decode $U_{\mathcal{W}_{\mathcal{U}}, \mathcal{T}_1, \mathcal{T}_2}$ from $X_{\mathcal{T}_2}$. Indeed, there exists an $\mathcal{S} \subseteq ([\Lambda] \setminus \mathcal{T}_2)$ with $|\mathcal{S}| = r$ such that $\mathcal{S} = (\mathcal{U}  \cup \mathcal{T}_1 ) \setminus \mathcal{T}_2$. For such $\mathcal{S}$, the corresponding coded message in $X_{\mathcal{T}_2}$ is 
\begin{align}
    \bigoplus_{\mathcal{R} \subseteq (\mathcal{S} \cup \mathcal{T}_2): |\mathcal{R}| = \alpha, \mathcal{R} \neq \mathcal{T}_2} U_{\mathcal{W}_{\mathcal{R}}, (\mathcal{S} \cup \mathcal{T}_2) \setminus \mathcal{R}, \mathcal{T}_2} & = \bigoplus_{\mathcal{R} \subseteq (\mathcal{U}  \cup \mathcal{T}_1) : |\mathcal{R}| = \alpha, \mathcal{R} \neq \mathcal{T}_2} U_{\mathcal{W}_{\mathcal{R}}, (\mathcal{U}  \cup \mathcal{T}_1 ) \setminus \mathcal{R}, \mathcal{T}_2} \\
    & = U_{\mathcal{W}_{\mathcal{U}}, \mathcal{T}_1, \mathcal{T}_2} \oplus \underbrace{\bigoplus_{\mathcal{R} \subseteq (\mathcal{U}  \cup \mathcal{T}_1) : |\mathcal{R}| = \alpha, \mathcal{R} \neq \mathcal{T}_2, \mathcal{R} \neq \mathcal{U}} U_{\mathcal{W}_{\mathcal{R}}, (\mathcal{U}  \cup \mathcal{T}_1 ) \setminus \mathcal{R}, \mathcal{T}_2}}_{\textnormal{interference}}.
\end{align}
Notice that reducer $\mathcal{U}$ can cancel the interference term by using the intermediate values retrieved from mappers in $\mathcal{U}$, so it can correctly decode $U_{\mathcal{W}_{\mathcal{U}}, \mathcal{T}_1, \mathcal{T}_2}$. By following the same rationale for each $\mathcal{T}_2 \subseteq (\mathcal{U} \cup \mathcal{T}_1)$ with $|\mathcal{T}_2| = \alpha$ and $\mathcal{T}_2 \neq \mathcal{U}$, we can conclude that reducer $\mathcal{U}$ can correctly recover $U_{\mathcal{W}_{\mathcal{U}}, \mathcal{T}_1}$ and can do so for each $\mathcal{T}_1 \subseteq ([\Lambda] \setminus \mathcal{U})$, completely recovering all the $\eta_2 \eta_1 \binom{\Lambda - \alpha}{r}$ missing intermediate values. The same holds for any other $\mathcal{U} \in [\Lambda]_\alpha$, and so we can conclude that each reducer is able to recover from the multicast messages of other reducers all the missing intermediate values.

\subsection{Communication Load}

The communication load guaranteed by the coded scheme described above is given by
\begin{align}
    L_{\textnormal{UB}}(r) & = \frac{\sum_{\mathcal{U} \in [\Lambda]_{\alpha}}|X_{\mathcal{U}}|}{QNT} \\
                        & = \frac{\binom{\Lambda}{\alpha} \eta_2 \eta_1 \binom{\Lambda - \alpha}{r} T/\left( \binom{r + \alpha}{r} - 1 \right)}{Q\eta_1\binom{\Lambda}{r}T} \\
                        & = \frac{\binom{\Lambda - \alpha}{r}}{\binom{\Lambda}{r}\left( \binom{r + \alpha}{r} - 1 \right)}
\end{align}
for each $r \in [\Lambda - \alpha + 1]$. Notice that the lower convex envelope of the achievable points $\{(r, L_{\textnormal{LB}}(r)) : r \in [\Lambda - \alpha + 1]\}$ is achievable by adopting the memory-sharing strategy presented in~\cite{Li2018FundamentalTradeoffComputation}. The proof is concluded. \qed

\section{Proof of Converse Bound in \texorpdfstring{\Cref{thm: Converse Bound Without Download Cost}}{Theorem~\ref{thm: Converse Bound Without Download Cost}}}\label{sec: Converse Without Download Cost Proof}

We begin the proof by introducing some useful notation. For $q \in [Q]$ and $n \in [N]$, we let $V_{q, n}$ be an i.i.d. random variable and we let $v_{q, n}$ be the realization of $V_{q, n}$. Then, we define
\begin{align}
    D_{\mathcal{U}} & \coloneqq \{V_{q, n} : q \in \mathcal{W}_{\mathcal{U}}, n \in [N]\} \\
    C_{\mathcal{U}} & \coloneqq \{V_{q, n} : q \in [Q], w_n \in \mathcal{M}_{\mathcal{U}}\} \\
    Y_{\mathcal{U}} & \coloneqq (D_{\mathcal{U}}, C_{\mathcal{U}}).
\end{align}
Recalling that we denote by $X_{\mathcal{U}}$ the multicast message transmitted by reducer $\mathcal{U} \in [\Lambda]_{\alpha}$, the equation
\begin{equation}
    H(X_{\mathcal{U}} \mid C_{\mathcal{U}}) = 0
\end{equation}
holds, since $X_{\mathcal{U}}$ is a function of the intermediate values retrieved by reducer $\mathcal{U}$. Moreover, for any map-shuffle-reduce scheme, each reducer $\mathcal{U} \in [\Lambda]_{\alpha}$ has to be able to correctly recover all the intermediate values $D_{\mathcal{U}}$ given the transmissions of all reducers $X_{[\Lambda]_\alpha} \coloneqq (X_{\mathcal{U}} : \mathcal{U} \in [\Lambda]_\alpha)$ and given the IVs $C_{\mathcal{U}}$ computed by the mappers in $\mathcal{U}$. Thus, the equation
\begin{equation}
    H(D_{\mathcal{U}} \mid X_{[\Lambda]_\alpha}, C_{\mathcal{U}}) = 0
\end{equation}
holds for each $\mathcal{U} \in [\Lambda]_{\alpha}$. 

\subsection{Lower Bound for a Given File Assignment}

For a given file assignment denoted by $\mathcal{M} \coloneqq (\mathcal{M}_1, \dots, \mathcal{M}_\Lambda)$, we let $L_{\mathcal{M}}$ be the corresponding communication load under this assignment $\mathcal{M}$. Then, we provide a lower bound on $L_{\mathcal{M}}$ for any given file assignment in the following lemma.

\begin{lemma}\label{lem: Lower Bound Lemma}
    Consider a specific file assignment $\mathcal{M} = (\mathcal{M}_1, \dots, \mathcal{M}_{\Lambda})$. Let $\bm{c} = (c_1, \dots, c_\Lambda)$ be a permutation of the set $[\Lambda]$ and define
    \begin{align}
         \mathcal{D}_{i} & \coloneqq (D_{\mathcal{U}^i} : \mathcal{U}^i \subseteq \{c_1, \dots, c_i\}, |\mathcal{U}^i| = \alpha, c_i \in \mathcal{U}^i) \\
         \mathcal{C}_{i} & \coloneqq (C_{\mathcal{U}^i} : \mathcal{U}^i \subseteq \{c_1, \dots, c_i\}, |\mathcal{U}^i| = \alpha, c_i \in \mathcal{U}^i) \\
         \mathcal{Y}_{i - 1} & \coloneqq (Y_{\mathcal{U}^{j}} : j \in [\alpha : i - 1], \mathcal{U}^{j} \subseteq \{c_1, \dots, c_{j}\}, |\mathcal{U}^{j}| = \alpha, c_{j} \in \mathcal{U}^{j})
    \end{align}
    for each $i \in [\alpha : \Lambda]$. Then, the communication load is lower bounded by
    \begin{equation}
        L_{\mathcal{M}} \geq \frac{1}{QNT} \sum_{i \in [\alpha : \Lambda]}H(\mathcal{D}_{i} \mid \mathcal{C}_{i}, \mathcal{Y}_{i - 1}).
    \end{equation}
\end{lemma}

\begin{proof}
    The proof is described in \refappendix{sec: Lower Bound Lemma Proof}.
\end{proof}

It is perhaps interesting to highlight that the lemma above manages to combine relatively divergent ideas from~\cite[Lemma 2]{Woolsey2021CombinatorialDesignCascaded} and~\cite[Lemma 2]{Brunero2021FundamentalLimitsCombinatorial}. On one hand, the proof of \Cref{lem: Lower Bound Lemma} is based on the iterative argument from the proof of~\cite[Lemma 2]{Woolsey2021CombinatorialDesignCascaded}, where the authors built a sequence of entropy-based bounds by iteratively picking computing nodes without ordering them according to some specific permutations. On the other hand, since in our case we wish to keep into account the multi-access nature of our MADC system, the proof of \Cref{lem: Lower Bound Lemma} adapts the entropy-based approach from~\cite{Woolsey2021CombinatorialDesignCascaded} by iteratively selecting the reducers according to some properly chosen permutations. The purpose of selecting reducers according to some proper permutations is that of constructing a tighter sequence of entropy-based bounds. The properly chosen permutations are inspired by~\cite[Lemma 2]{Brunero2021FundamentalLimitsCombinatorial}, which was used to successfully develop a tight converse bound for the multi-access coded caching problem with combinatorial topology.

Now, we proceed with the proof. Denote by $\tilde{a}^{\mathcal{T}}$ the number of files which are mapped exclusively by the mappers in $\mathcal{T}$ for some $\mathcal{T} \subseteq [\Lambda]$. As each reducer $\mathcal{U} \in [\Lambda]_{\alpha}$ does not have access to the intermediate values of all those files that are not mapped by the mappers in $\mathcal{U}$, the term $\tilde{a}^{\mathcal{T}}$ represents the number of files whose intermediate values are required by each reducer $\mathcal{U} \in [\Lambda]_{\alpha}$ that does not have access to the mappers in $\mathcal{T}$, i.e., each reducer $\mathcal{U} \in [\Lambda]_{\alpha}$ such that $\mathcal{U} \cap \mathcal{T} = 0$ or, equivalently, each reducer $\mathcal{U} \subseteq ([\Lambda] \setminus \mathcal{T})$ such that $|\mathcal{U}| = \alpha$. Taking advantage of the independence of the intermediate values and recalling that each reducer computes $\eta_2$ disjoint output functions, from \Cref{lem: Lower Bound Lemma} and for a given permutation $\bm{c} = (c_1, \dots, c_{\Lambda})$ of the set $[\Lambda]$, we can further write
\begin{align}
    L_{\mathcal{M}} & \geq \frac{1}{QNT}\sum_{i \in [\alpha : \Lambda]} H(\mathcal{D}_{i} \mid \mathcal{C}_{i}, \mathcal{Y}_{i - 1}) \\
             & = \frac{1}{QNT}\sum_{i \in [\alpha : \Lambda]} \sum_{\mathcal{U}^i \subseteq \{c_1, \dots, c_i\} : |\mathcal{U}^i| = \alpha, c_i \in \mathcal{U}^i} H(D_{\mathcal{U}^i} \mid \mathcal{C}_i, \mathcal{Y}_{i - 1}) \\
             & = \frac{1}{QNT} \sum_{i \in [\alpha : \Lambda]} \sum_{\mathcal{U}^i \subseteq \{c_1, \dots, c_i\} : |\mathcal{U}^i| = \alpha, c_i \in \mathcal{U}^i} \sum_{\mathcal{T} \subseteq [\Lambda] \setminus \{c_1, \dots, c_i\}} \tilde{a}^{\mathcal{T}} \eta_2 T \\
             & = \frac{1}{KN} \sum_{i \in [\alpha : \Lambda]} \sum_{\mathcal{U}^i \subseteq \{c_1, \dots, c_i\} : |\mathcal{U}^i| = \alpha, c_i \in \mathcal{U}^i} \sum_{\mathcal{T} \subseteq [\Lambda] \setminus \{c_1, \dots, c_i\}} \tilde{a}^{\mathcal{T}}.
\end{align}
If we build a bound as the one in \Cref{lem: Lower Bound Lemma} for each permutation of the set $[\Lambda]$ and we sum up all these bounds together, we obtain the expression
\begin{equation}\label{eqn: Lower Bound File Assignment}
    L_{\mathcal{M}} \geq \frac{1}{KN\Lambda !} \sum_{\bm{c} \in S_{\Lambda}} \sum_{i \in [\alpha : \Lambda]} \sum_{\mathcal{U}^i \subseteq \{c_1, \dots, c_i\} : |\mathcal{U}^i| = \alpha, c_i \in \mathcal{U}^i} \sum_{\mathcal{T} \subseteq [\Lambda] \setminus \{c_1, \dots, c_i\}} \tilde{a}^{\mathcal{T}}
\end{equation}
where we recall that $S_{\Lambda}$ represents the group of all permutations of $[\Lambda]$. Our goal now is to simplify this expression and we start doing so by counting how many times each term $\tilde{a}^{\mathcal{T}}$ appears in the RHS of~\eqref{eqn: Lower Bound File Assignment} for any fixed $\mathcal{T} \subseteq [\Lambda]$ with $|\mathcal{T}| = j$ and $j \in [\Lambda]$.

First, we focus on some reducer $\mathcal{U} \subseteq ([\Lambda] \setminus \mathcal{T})$ with $|\mathcal{U}| = \alpha$. We can see that $\tilde{a}^{\mathcal{T}}$ appears in the RHS of~\eqref{eqn: Lower Bound File Assignment} for all those permutations in $S_{\Lambda}$ for which $\mathcal{U} = \mathcal{U}^{i}$ for some $i \in [\alpha : \Lambda]$ such that $\mathcal{U}^{i} \subseteq \{c_1, \dots, c_i\}$ with $|\mathcal{U}^i| = \alpha$ and $c_i \in \mathcal{U}^i$, and such that $\mathcal{T} \subseteq ([\Lambda] \setminus \{c_1, \dots, c_i\})$. Denoting by $\mathcal{P}_{\mathcal{U}, \mathcal{T}}$ the set of such permutations, we can see that
\begin{equation}
    |\mathcal{P}_{\mathcal{U}, \mathcal{T}}| = \alpha! j! (\Lambda - \alpha - j)! \binom{\Lambda}{\alpha + j}.
\end{equation}
The same reasoning applies to any reducer $\mathcal{U} \in [\Lambda]_{\alpha}$ for which $\mathcal{U} \cap \mathcal{T} = \emptyset$. As a consequence, the term $\tilde{a}^{\mathcal{T}}$ appears in the RHS of~\eqref{eqn: Lower Bound File Assignment} a total of 
\begin{equation}
    \sum_{\mathcal{U} \in [\Lambda]_{\alpha} : \mathcal{U} \cap \mathcal{T} = \emptyset} |\mathcal{P}_{\mathcal{U}, \mathcal{T}}| = \binom{\Lambda - j}{\alpha} \alpha! j! (\Lambda - \alpha - j)! \binom{\Lambda}{\alpha + j}
\end{equation}
times. The same rationale holds for any $\tilde{a}^{\mathcal{T}}$ where $\mathcal{T} \subseteq [\Lambda]$ and $|\mathcal{T}| = j$ with $j \in [\Lambda]$. Consequently, we can rewrite the expression in~\eqref{eqn: Lower Bound File Assignment} as
\begin{align}
    L_{\mathcal{M}} & \geq \frac{1}{KN\Lambda !} \sum_{\bm{c} \in S_{\Lambda}} \sum_{i \in [\alpha : \Lambda]} \sum_{\mathcal{U}^i \subseteq \{c_1, \dots, c_i\} : |\mathcal{U}^i| = \alpha, c_i \in \mathcal{U}^i} \sum_{\mathcal{T} \subseteq [\Lambda] \setminus \{c_1, \dots, c_i\}} \tilde{a}^{\mathcal{T}} \\
             & = \frac{1}{KN\Lambda !} \sum_{j \in [\Lambda]} \sum_{\mathcal{T} \subseteq [\Lambda] : |\mathcal{T}| = j} \binom{\Lambda - j}{\alpha}\alpha! j! (\Lambda - \alpha - j)! \binom{\Lambda}{\alpha + j} \tilde{a}^{\mathcal{T}} \\
             & = \frac{1}{KN} \sum_{j \in [\Lambda]} \frac{\binom{\Lambda}{\alpha + j}}{\binom{\Lambda}{j}} \sum_{\mathcal{T} \subseteq [\Lambda] : |\mathcal{T}| = j} \tilde{a}^{\mathcal{T}} \\
             & = \frac{1}{K} \sum_{j \in [\Lambda]} \frac{\binom{\Lambda}{\alpha + j}}{\binom{\Lambda}{j}} \frac{\tilde{a}^j_{\mathcal{M}}}{N}
\end{align}
where $\tilde{a}^j_{\mathcal{M}} \coloneqq \sum_{\mathcal{T} \subseteq [\Lambda] : |\mathcal{T}| = j} \tilde{a}^{\mathcal{T}}$ is defined as the total number of files which are mapped by exactly $j$ map nodes under this particular file assignment $\mathcal{M}$.

For any given file assignment $\mathcal{M}$ and for any given computation load $r \in [K]$, the fact that $|\mathcal{M}_1| + \dots + |\mathcal{M}_{\Lambda}| \leq rN$ also implies that $\tilde{a}^j_{\mathcal{M}} \geq 0$ for each $j \in [\Lambda]$, as well as implies that $\sum_{j \in [\Lambda]}\tilde{a}^j_{\mathcal{M}} = N$ and that $\sum_{j \in [\Lambda]}j \tilde{a}^j_{\mathcal{M}} \leq rN$. Thus, we can further lower bound the above using Jensen's inequality and the fact that $\binom{\Lambda}{\alpha + j}/\binom{\Lambda}{j}$ is convex and decreasing\footnote{This was already proved in the proof of~\cite[Lemma 3]{Brunero2021FundamentalLimitsCombinatorial} by writing down each combinatorial coefficient in $\binom{\Lambda}{\alpha + j}/\binom{\Lambda}{j}$ as a finite product and using then the general Leibniz rule to show that its second derivative is non-negative.} in $j$. Hence, we can write
\begin{align}
    L_{\mathcal{M}} & \geq \frac{1}{K} \sum_{j \in [\Lambda]} \frac{\binom{\Lambda}{\alpha + j}}{\binom{\Lambda}{j}} \frac{\tilde{a}^j_{\mathcal{M}}}{N} \\
                            & \geq \frac{1}{K} \frac{\binom{\Lambda}{\alpha + r}}{\binom{\Lambda}{r}} \label{eqn: Storage Constraint} \\
                            & = \frac{\binom{\Lambda}{\alpha + r}}{\binom{\Lambda}{\alpha}\binom{\Lambda}{r}} \label{eqn: Refined Lower Bound}
\end{align}
where~\eqref{eqn: Storage Constraint} holds due to the storage constraint $\sum_{j \in [\Lambda]}j \tilde{a}^j_{\mathcal{M}} \leq rN$.

\subsection{Lower Bound Over All Possible File Assignments}

To obtain the bound in \Cref{thm: Converse Bound Without Download Cost}, we are looking for the smallest $L_{\mathcal{M}}$ across all file assignments $\mathcal{M}$ such that $|\mathcal{M}_1| + \dots + |\mathcal{M}_{\Lambda}| \leq rN$, that is we are looking for
\begin{equation}
    L^{\star}(r) \geq \inf_{\mathcal{M} : |\mathcal{M}_1| + \dots + |\mathcal{M}_{\Lambda}| \leq rN} L_{\mathcal{M}}.
\end{equation}
Given that~\eqref{eqn: Refined Lower Bound} is independent of the file assignment $\mathcal{M}$ and lower bounds $L_{\mathcal{M}}$ for any $\mathcal{M}$ such that $|\mathcal{M}_1| + \dots + |\mathcal{M}_{\Lambda}| \leq rN$, we can further write
\begin{align}
    L^{\star}(r) & \geq \inf_{\mathcal{M} : |\mathcal{M}_1| + \dots + |\mathcal{M}_{\Lambda}| \leq rN} L_{\mathcal{M}} \\
                 & \geq \inf_{\mathcal{M} : |\mathcal{M}_1| + \dots + |\mathcal{M}_{\Lambda}| \leq rN} \frac{\binom{\Lambda}{\alpha + r}}{\binom{\Lambda}{\alpha}\binom{\Lambda}{r}} \\
                 & = \frac{\binom{\Lambda}{\alpha + r}}{\binom{\Lambda}{\alpha}\binom{\Lambda}{r}} \\
                 & = L_{\textnormal{LB}}(r).
\end{align}
Notice that the bound $L_{\textnormal{LB}}(r)$ can be extended to include also the non-integer values of $r$ as described in~\cite{Li2018FundamentalTradeoffComputation}. This concludes the proof. \qed

\section{Proof of Achievable Bound in \texorpdfstring{\Cref{thm: Achievable Bound With Download Cost}}{Theorem~\ref{thm: Achievable Bound With Download Cost}}}\label{sec: Achievability With Download Cost Proof}

As we mentioned in the statement of \Cref{thm: Achievable Bound With Download Cost}, the coded scheme depends on the solution of the linear program in~\eqref{eqn: Linear Program}. Hence, the first step is to evaluate the optimal solution\footnote{The linear program in~\eqref{eqn: Linear Program} is not infeasible nor unbounded. Hence, it admits an optimal solution.} $\tilde{\bm{a}}_{\star} = (\tilde{a}^{1}_{\star}, \dots, \tilde{a}^{\Lambda}_{\star})$. Next, we partition the input database in $\Lambda$ parts, where we denote by $\mathcal{L}_{j}$ the $j$-th part, which contains $|\mathcal{L}_{j}| = \tilde{a}^{j}_{\star}$ files for each $j \in [\Lambda]$. Then, each part $j \in [\Lambda]$ of the database is split in $\binom{\Lambda}{j}$ batches containing $\eta_{j}$ files each for some $\eta_{j} \in \mathbb{N}$, so that $\tilde{a}^{j}_{\star} = \eta_{j} \binom{\Lambda}{j}$ for each $j \in [\Lambda]$. This implies
\begin{align}
    \{w_1, \dots, w_N\} & = \bigcup_{j \in [\Lambda]} \mathcal{L}_{j} \\
                             & = \bigcup_{j \in [\Lambda]} \bigcup_{\mathcal{T}_1 \subseteq [\Lambda] : |\mathcal{T}_{1}| = j} \mathcal{B}_{j, \mathcal{T}_1}
\end{align}
where we denote by $\mathcal{B}_{j, \mathcal{T}_{1}}$ the batch containing $\eta_{j}$ files associated with the label $\mathcal{T}_{1}$. Then, mapper $\lambda \in [\Lambda]$ is assigned all batches $\mathcal{B}_{j, \mathcal{T}_{1}}$ having $\lambda \in \mathcal{T}_{1}$ for each $j \in [\Lambda]$, which implies
\begin{equation}
    \mathcal{M}_{\lambda} = \{\mathcal{B}_{j, \mathcal{T}_{1}} : j \in [\Lambda], \mathcal{T}_{1} \subseteq [\Lambda], |\mathcal{T}_{1}| = j, \lambda \in \mathcal{T}_{1}\}.
\end{equation}
The computation load constraint is satisfied, since we have
\begin{align}
    \frac{\sum_{\lambda \in [\Lambda]} |\mathcal{M}_{\lambda}|}{N} & = \frac{\Lambda \sum_{j \in [\Lambda]} \eta_{j} \binom{\Lambda - 1}{j - 1}}{N} \\
    & = \frac{\sum_{j \in [\Lambda]} j \eta_{j} \binom{\Lambda}{j}}{N} \\
    & = \frac{\sum_{j \in [\Lambda]} j \tilde{a}^{j}_{\star}}{N} \leq r
\end{align}
where the last inequality holds under the constraint in~\eqref{eqn: Computation Load Constraint}.

Our goal is to provide an achievable scheme for the max-link communication load. Recalling that we denote by $L$ and $J$ the communication load and the download cost, respectively, we will have
\begin{equation}
    L^{\star}_{\text{max-link}}(r) \leq L_{\text{max-link}, \text{UB}}(r) = \max\left(L, D\right).
\end{equation}

\subsection{Communication Load}

For what concerns the communication load, we can take advantage of the achievable scheme described in \Cref{sec: Achievability Without Download Cost Proof}. Simply, the scheme in \Cref{sec: Achievability Without Download Cost Proof} is applied $\Lambda$ times, one time per partition $\mathcal{L}_{j}$ which is considered as an independent input database. If we denote by $L_{j}$ the communication load when we focus on the part $\mathcal{L}_{j}$, we have that $L_{j}$ is given by
\begin{equation}
    L_{j} = \frac{\binom{\Lambda - \alpha}{j}}{\binom{\Lambda}{j}\left(\binom{j + \alpha}{j} - 1\right)}\frac{\tilde{a}^{j}_{\star}}{N}
\end{equation}
for each $j \in [\Lambda]$. Hence, the overall communication load $L$ is given by
\begin{equation}
    L = \sum_{j \in [\Lambda]}L_{j} = \sum_{j \in [\Lambda]} \frac{\binom{\Lambda - \alpha}{j}}{\binom{\Lambda}{j}\left(\binom{j + \alpha}{j} - 1\right)}\frac{\tilde{a}^{j}_{\star}}{N}.
\end{equation}

\subsection{Download Cost}

We remind that the download cost is defined as
\begin{equation}
    J = \max_{\lambda \in [\Lambda]} \max_{\mathcal{U} \in [\Lambda]_{\alpha} : \lambda \in \mathcal{U}} \frac{R^{\mathcal{U}}_{\lambda}}{QNT}
\end{equation}
where $R^{\mathcal{U}}_{\lambda}$ represents the number of bits which are sent from mapper $\lambda$ to reducer $\mathcal{U}$. This quantity is minimized if the number of bits transmitted over each link connecting a mapper to a reducer is the same. This can be accomplished as follows.

Consider a reducer $\mathcal{U} \in [\Lambda]_{\alpha}$ and a mapper $\lambda \in \mathcal{U}$. According to the file assignment above, mapper $\lambda$ computes the IVs in the set $\mathcal{V}_{\lambda} = \{v_{q, n} : q \in [Q], w_{n} \in \mathcal{M}_{\lambda}\}$. The set $\mathcal{V}_{\lambda}$ can equivalently be written as follows
\begin{equation}
    \mathcal{V}_{\lambda} = \{\mathcal{V}_{\lambda, \mathcal{S}} : i \in [\alpha], \mathcal{S} \subseteq (\mathcal{U} \setminus \{\lambda\}), |\mathcal{S}| = i - 1\}
\end{equation}
where $\mathcal{V}_{\lambda, \mathcal{S}}$ is defined as
\begin{equation}
    \mathcal{V}_{\lambda, \mathcal{S}} \coloneqq \left\{v_{q, n} : q \in [Q], w_{n} \in \mathcal{M}^{\lambda \cup \mathcal{S}}\right\}
\end{equation}
and where $\mathcal{M}^{\lambda \cup \mathcal{S}} \coloneqq \bigcap_{s \in (\lambda \cup \mathcal{S})} \mathcal{M}_{s}$. This simply says that the set $\mathcal{V}_{\lambda, \mathcal{S}}$ contains the IVs which are mapped by mapper $\lambda$ and the $(i - 1)$ mappers in $\mathcal{S}$. Hence, if we evenly split $\mathcal{V}_{\lambda, \mathcal{S}}$ in $i$ segments as follows
\begin{equation}
    \mathcal{V}_{\lambda, \mathcal{S}} = \left(\mathcal{V}_{\lambda, \mathcal{S}, s} : s \in (\lambda \cup \mathcal{S}) \right)
\end{equation}
we simply let mapper $\lambda$ send $\mathcal{V}_{\lambda, \mathcal{S}, \lambda}$. This implies
\begin{align}
    R^{\mathcal{U}}_{\lambda} & = \sum_{i \in [\alpha]} \sum_{\mathcal{S} \subseteq (\mathcal{U} \setminus \{\lambda\}) : |\mathcal{S}| = i - 1 } |\mathcal{V}_{\lambda, \mathcal{S}, \lambda}| \\
                              & = \sum_{i \in [\alpha]} \sum_{\mathcal{S} \subseteq (\mathcal{U} \setminus \{\lambda\}) : |\mathcal{S}| = i - 1 } \frac{|\mathcal{V}_{\lambda, \mathcal{S}}|}{i} \\
                              & = \sum_{i \in [\alpha]} \sum_{\mathcal{S} \subseteq (\mathcal{U} \setminus \{\lambda\}) : |\mathcal{S}| = i - 1 } \sum_{j \in [\Lambda]} \frac{\eta_{j} \binom{\Lambda - \alpha}{j - i}Q T}{i} \\
                              & = \sum_{j \in [\Lambda]} \sum_{i \in [\alpha]} \frac{\eta_{j} \binom{\Lambda - \alpha}{j - i}Q T}{i} \binom{\alpha - 1}{i - 1}
\end{align}
for each $\lambda \in [\Lambda]$ and $\mathcal{U} \in [\Lambda]_{\alpha}$ with $\lambda \in \mathcal{U}$. Hence, we can further write
\begin{align}
    J & = \max_{\lambda \in [\Lambda]} \max_{\mathcal{U} \in [\Lambda]_{\alpha} : \lambda \in \mathcal{U}} \frac{R^{\mathcal{U}}_{\lambda}}{QNT} \\
      & = \frac{1}{QNT} \sum_{j \in [\Lambda]} \sum_{i \in [\alpha]} \frac{\eta_{j} \binom{\Lambda - \alpha}{j - i}Q T}{i} \binom{\alpha - 1}{i - 1} \\
      & = \sum_{j \in [\Lambda]} \sum_{i \in [\alpha]} \frac{\binom{\Lambda - \alpha}{j - i} \binom{\alpha - 1}{i - 1}}{i\binom{\Lambda}{j}} \frac{\tilde{a}^{j}_{\star}}{N}
\end{align}
recalling that $\tilde{a}^{j}_{\star} = \eta_{j} \binom{\Lambda}{j}$ for each $j \in [\Lambda]$. Further, the following lemma holds.

\begin{lemma}\label{lem: Download Cost Lemma}
    For any non-negative integers $\Lambda$, $\alpha$ and $j$, we have
    \begin{equation}
        \sum_{i \in [\alpha]} \frac{\binom{\Lambda - \alpha}{j - i} \binom{\alpha - 1}{i - 1}}{i\binom{\Lambda}{j}} = \frac{\binom{\Lambda}{j} - \binom{\Lambda - \alpha}{j}}{\alpha \binom{\Lambda}{j}}.
    \end{equation}
\end{lemma}
\begin{proof}
    The proof is described in \refappendix{sec: Download Cost Lemma Proof}.
\end{proof}

As a consequence, the download cost $J$ is equivalently given by
\begin{align}
    J & = \sum_{j \in [\Lambda]} \frac{\binom{\Lambda}{j} - \binom{\Lambda - \alpha}{j}}{\alpha \binom{\Lambda}{j}}\frac{\tilde{a}^{j}_{\star}}{N} \\
      & = \sum_{j \in [\Lambda]} \frac{\binom{\Lambda}{\alpha} - \binom{\Lambda - j}{\alpha}}{\alpha \binom{\Lambda}{\alpha}}\frac{\tilde{a}^{j}_{\star}}{N}.
\end{align}

\subsection{Max-Link Communication Load}
Since we have now the expressions for both $L$ and $J$, we can write explicitly the achievable max-link communication load as follows
\begin{align}
    L_{\text{max-link}, \text{UB}}(r) & = \max\left(L, D\right) \\
                                      & = \max\left(\sum_{j \in [\Lambda]} \frac{\binom{\Lambda - \alpha}{j}}{\binom{\Lambda}{j}\left(\binom{j + \alpha}{j} - 1\right)}\frac{\tilde{a}^{j}_{\star}}{N}, \sum_{j \in [\Lambda]} \frac{\binom{\Lambda}{\alpha} - \binom{\Lambda - j}{\alpha}}{\alpha \binom{\Lambda}{\alpha}}\frac{\tilde{a}^{j}_{\star}}{N}\right).
\end{align}
The expression above coincides with the achievable expression in \Cref{thm: Achievable Bound With Download Cost}. The proof is concluded. \qed

\section{Proof of Converse Bound in \texorpdfstring{\Cref{thm: Converse Bound With Download Cost}}{Theorem~\ref{thm: Converse Bound With Download Cost}}}\label{sec: Converse With Download Cost Proof}

We quickly recall that $L_{\mathcal{M}}$ denotes the communication load under the file assignment $\mathcal{M} = (\mathcal{M}_1, \dots, \mathcal{M}_\Lambda)$. Then, we denote by $J_{\mathcal{M}}$ and by $L_{\text{max-link}, \mathcal{M}}(r) = \max\{L_{\mathcal{M}}, J_{\mathcal{M}}\}$ the download cost and the max-link communication load, respectively, under file assignment $\mathcal{M}$.

\subsection{Lower Bound for a Given File Assignment}

From \Cref{sec: Converse Without Download Cost Proof} we know that the inequality
\begin{equation}
    L_{\mathcal{M}} \geq \sum_{j \in [\Lambda]} \frac{\binom{\Lambda}{\alpha + j}}{\binom{\Lambda}{\alpha}\binom{\Lambda}{j}} \frac{\tilde{a}^j_{\mathcal{M}}}{N}
\end{equation}
holds. Thus, we can write
\begin{align}
    L_{\text{max-link}, \mathcal{M}} (r) & = \max\{L_{\mathcal{M}}, J_{\mathcal{M}}\} \\
                                     & \geq \max\left\{\sum_{j \in [\Lambda]} \frac{\binom{\Lambda}{\alpha + j}}{\binom{\Lambda}{\alpha}\binom{\Lambda}{j}} \frac{\tilde{a}^j_{\mathcal{M}}}{N}, J_{\mathcal{M}}\right\}.
\end{align}
In the following, we wish to develop a lower bound on $J_{\mathcal{M}}$. Starting from the definition of the download cost, we have
\begin{subequations}
    \begin{align}
        J_{\mathcal{M}} & = \max_{\lambda \in [\Lambda]} \max_{\mathcal{U} \in [\Lambda]_{\alpha}: \lambda \in \mathcal{U}} \frac{R^{\mathcal{U}}_{\lambda}}{QNT} \\
                        & \geq \max_{\lambda \in [\Lambda]} \frac{1}{\binom{\Lambda - 1}{\alpha - 1}QNT} \sum_{\mathcal{U} \in [\Lambda]_{\alpha}: \lambda \in \mathcal{U}} R^{\mathcal{U}}_{\lambda} \\
                        & \geq \frac{1}{\Lambda\binom{\Lambda - 1}{\alpha - 1}QNT} \sum_{\lambda \in [\Lambda]} \sum_{\mathcal{U} \in [\Lambda]_{\alpha}: \lambda \in \mathcal{U}} R^{\mathcal{U}}_{\lambda} \\
                        & = \frac{1}{\alpha \binom{\Lambda}{\alpha}QNT} \sum_{\mathcal{U} \in [\Lambda]_{\alpha} } \sum_{\lambda \in \mathcal{U}} R^{\mathcal{U}}_{\lambda} \\
                        & = \frac{1}{\alpha \binom{\Lambda}{\alpha}QNT} \sum_{\mathcal{U} \in [\Lambda]_{\alpha} } R^{\mathcal{U}}
    \end{align}
\end{subequations}
where $R^{\mathcal{U}}$ is defined as
\begin{equation}
    R^{\mathcal{U}} \coloneqq \sum_{\lambda \in \mathcal{U}} R^{\mathcal{U}}_{\lambda}
\end{equation}
to represent the overall number of bits received by reducer $\mathcal{U} \in [\Lambda]_{\alpha}$. Now, since each reducer $\mathcal{U}$ is expected to receive all the IVs mapped by the mappers in $\mathcal{U}$, we have
\begin{equation}
    R^{\mathcal{U}} \geq H(C_{\mathcal{U}})
\end{equation}
where we recall that $C_{\mathcal{U}} = \{V_{q, n} : q \in [Q], w_n \in \mathcal{M}_{\mathcal{U}}\}$ from \Cref{sec: Converse Without Download Cost Proof}. This means that we can further write
\begin{subequations}
    \begin{align}
        J_{\mathcal{M}} & \geq \frac{1}{\alpha \binom{\Lambda}{\alpha}QNT} \sum_{\mathcal{U} \in [\Lambda]_{\alpha} } R^{\mathcal{U}} \\
                        & \geq \frac{1}{\alpha \binom{\Lambda}{\alpha}QNT} \sum_{\mathcal{U} \in [\Lambda]_{\alpha} } H(C_{\mathcal{U}}) \\
                        & = \frac{1}{\alpha \binom{\Lambda}{\alpha}QNT} \sum_{\mathcal{U} \in [\Lambda]_{\alpha} } \sum_{\mathcal{T} \subseteq [\Lambda] : \mathcal{T} \cap \mathcal{U} \neq \emptyset} \tilde{a}^{\mathcal{T}} Q T \\
                        & = \frac{1}{\alpha \binom{\Lambda}{\alpha}N} \sum_{j \in [\Lambda]} \sum_{\mathcal{T} \subseteq [\Lambda] : |\mathcal{T}| = j} \left( \binom{\Lambda}{\alpha} - \binom{\Lambda - j}{\alpha} \right) \tilde{a}^{\mathcal{T}} \\
                        & = \sum_{j \in [\Lambda]} \frac{\binom{\Lambda}{\alpha} - \binom{\Lambda - j}{\alpha}}{\alpha \binom{\Lambda}{\alpha}} \frac{\tilde{a}^{j}_{\mathcal{M}}}{N}
    \end{align}
\end{subequations}
recalling that $\tilde{a}^{j}_{\mathcal{M}} = \sum_{\mathcal{T} \subseteq [\Lambda] : |\mathcal{T}| = j}  \tilde{a}^{\mathcal{T}}$. To conclude, for a given file assignment $\mathcal{M}$, the max-link communication load is lower bounded as 
\begin{align}
    L_{\text{max-link}, \mathcal{M}} (r) & \geq \max\left\{\sum_{j \in [\Lambda]} \frac{\binom{\Lambda}{\alpha + j}}{\binom{\Lambda}{\alpha}\binom{\Lambda}{j}} \frac{\tilde{a}^j_{\mathcal{M}}}{N}, J_{\mathcal{M}}\right\} \\
                                         & \geq \max\left\{\sum_{j \in [\Lambda]} \frac{\binom{\Lambda}{\alpha + j}}{\binom{\Lambda}{\alpha}\binom{\Lambda}{j}} \frac{\tilde{a}^j_{\mathcal{M}}}{N}, \sum_{j \in [\Lambda]} \frac{\binom{\Lambda}{\alpha} - \binom{\Lambda - j}{\alpha}}{\alpha \binom{\Lambda}{\alpha}} \frac{\tilde{a}^{j}_{\mathcal{M}}}{N} \right\} \\
                                        & \geq \frac{1}{2}\sum_{j \in [\Lambda]} \left( \frac{\binom{\Lambda}{\alpha + j}}{\binom{\Lambda}{\alpha}\binom{\Lambda}{j}} + \frac{\binom{\Lambda}{\alpha} - \binom{\Lambda - j}{\alpha}}{\alpha \binom{\Lambda}{\alpha}} \right) \frac{\tilde{a}^j_{\mathcal{M}}}{N}.
\end{align}

\subsection{Lower Bound Over All Possible File Assignments}

Our aim is to develop a bound on the max-link communication load under any possible file assignment, namely, we are looking for the smallest $L_{\text{max-link}, \mathcal{M}}(r)$ across all file assignments $\mathcal{M}$ such that $|\mathcal{M}_1| + \dots + |\mathcal{M}_{\Lambda}| \leq rN$ for a given computation load $r \in [K]$. Since each file assignment $\mathcal{M}$ such that $|\mathcal{M}_1| + \dots + |\mathcal{M}_{\Lambda}| \leq rN$ also implies that $\tilde{a}^j_{\mathcal{M}} \geq 0$ for each $j \in [\Lambda]$, as well as implies that $\sum_{j \in [\Lambda]}\tilde{a}^j_{\mathcal{M}} = N$ and that $\sum_{j \in [\Lambda]}j \tilde{a}^j_{\mathcal{M}} \leq rN$, the max-link load $L^{\star}_{\text{max-link}}(r)$ is lower bounded by the solution to the following linear program
\begin{subequations}
    \begin{alignat}{2}
            & \min_{\tilde{\bm{a}}_{\mathcal{M}}} & \quad & \frac{1}{2} \sum_{j \in [\Lambda]} \left(\frac{\binom{\Lambda}{\alpha + j}}{\binom{\Lambda}{\alpha}\binom{\Lambda}{j}} + \frac{\binom{\Lambda}{\alpha} - \binom{\Lambda - j}{\alpha}}{\alpha \binom{\Lambda}{\alpha}} \right) \frac{\tilde{a}^{j}_{\mathcal{M}}}{N} \\
            & \textnormal{subject to}  &  & \tilde{a}^{j}_{\mathcal{M}} \geq 0, \quad \forall j \in [\Lambda] \\
            & & & \sum_{j \in [\Lambda]} \frac{\tilde{a}^{j}_{\mathcal{M}}}{N} = 1 \\
            & & & \sum_{j \in [\Lambda]} j\frac{\tilde{a}^{j}_{\mathcal{M}}}{N} \leq r
    \end{alignat}
\end{subequations}
where $\tilde{\bm{a}}_{\mathcal{M}} = (\tilde{a}^{1}_{\mathcal{M}}, \dots, \tilde{a}^{\Lambda}_{\mathcal{M}})$ is the control variable. The proof is concluded. \qed

\section{Conclusions}\label{sec: Conclusions}

In this work, we introduced multi-access distributed computing, a novel system model that generalizes the original CDC setting by considering mappers and reducers as distinct entities, and by considering each reducer to be connected to multiple mappers through a network topology. We focused on the MADC model with combinatorial topology, which implies $\Lambda$ mappers and $K = \binom{\Lambda}{\alpha}$ reducers, so that there is a reducer for any set of $\alpha$ mappers. Neglecting at first the download cost from mappers to reducers and so focusing only on the inter-reducer communication load, we proposed a novel coded scheme which, together with an information-theoretic converse, characterizes the optimal communication load within a constant multiplicative gap of $1.5$. Subsequently, we jointly considered the setting which keeps into account the download cost and for such scenario we characterized the optimal max-link communication load within a multiplicative factor of $4$. We point out that the here introduced achievable shuffling scheme --- which generalizes the original coded scheme in~\cite{Li2018FundamentalTradeoffComputation} (corresponding to the case $\alpha = 1$) --- offers also unparalleled coding gains. As an outcome of this gain, we have the interesting occurrence that our scheme guarantees a smaller communication load when $\alpha > 1$, capitalizing on the multi-access nature of the MADC model, even though the number of reducers is increased.

Interesting future directions could include the study of the here proposed MADC setting when mappers and reducers have heterogeneous computational resources. A careful study of other multi-access network topologies is also another challenging research direction. Reflecting a design freedom, the search for the best possible topology, for a given computation load, remains a very pertinent open problem in distributed computing.

\appendices

\section{Proof of \texorpdfstring{\Cref{cor: Multi-Access Degree Improvement}}{Corollary~\ref{cor: Multi-Access Degree Improvement}}}\label{cor: Multi-Access Degree Improvement Proof}

We wish to prove that, for fixed computation load $r$, the achievable performance in \Cref{thm: Achievable Bound Without Download Cost} decreases for increasing $\alpha$. To do so, it is enough to prove that
\begin{equation}
    s_{\alpha} > s_{\alpha + 1}
\end{equation}
where $s_{\alpha}$ is defined as
\begin{equation}
    s_{\alpha} \coloneqq \frac{\binom{\Lambda - \alpha}{r}}{\binom{\Lambda}{r}\left( \binom{r + \alpha}{r} - 1 \right)}.
\end{equation}
Toward this, we can see that
\begin{align}
    s_{\alpha + 1} & = \frac{\binom{\Lambda - \alpha - 1}{r}}{\binom{\Lambda}{r}\left( \binom{r + \alpha + 1}{r} - 1 \right)} \\
                   & = \frac{\frac{\Lambda - \alpha - r}{\Lambda - \alpha}\binom{\Lambda - \alpha}{r}}{\binom{\Lambda}{r}\left( \frac{r + \alpha + 1}{\alpha + 1} \binom{r + \alpha}{r} - 1 \right)} \\
                   & = \frac{\left(1 - \frac{r}{\Lambda - \alpha} \right)\binom{\Lambda - \alpha}{r}}{\binom{\Lambda}{r}\left( \left( 1 + \frac{r}{\alpha + 1} \right) \binom{r + \alpha}{r} - 1 \right)}.
\end{align}
Since $r/(\Lambda - \alpha) > 0$ and $r/(\alpha + 1) > 0$, it holds that $1 - r/(\Lambda - \alpha) < 1$ and $1 + r/(\alpha + 1) > 1$. Consequently, we can write
\begin{align}
    s_{\alpha + 1} & = \frac{\left(1 - \frac{r}{\Lambda - \alpha} \right)\binom{\Lambda - \alpha}{r}}{\binom{\Lambda}{r}\left( \left( 1 + \frac{r}{\alpha + 1} \right) \binom{r + \alpha}{r} - 1 \right)} \\
                   & < \frac{\binom{\Lambda - \alpha}{r}}{\binom{\Lambda}{r}\left(\binom{r + \alpha}{r} - 1 \right)} \\
                   & = s_{\alpha}
\end{align}
showing in this way that $s_{\alpha + 1} < s_{\alpha}$. This concludes the proof. \qed

\section{Proof of Order Optimality in \texorpdfstring{\Cref{thm: Order Optimality Without Download Cost}}{Theorem~\ref{thm: Order Optimality Without Download Cost}}}\label{sec: Order Optimality Without Download Cost Proof}

To prove the order optimality result in \Cref{thm: Order Optimality Without Download Cost}, we need to upper bound the ratio $L_{\textnormal{UB}}(r)/L^{\star}(r)$ for each $r \in [\Lambda - \alpha + 1]$. We start by noting that the following
\begin{align}
    \frac{L_{\textnormal{UB}}(r)}{L^{\star}(r)} & \leq \frac{L_{\textnormal{UB}}(r)}{L_{\textnormal{LB}}(r)} \\
                                             & = \frac{\binom{\Lambda - \alpha}{r}}{\cancel{\binom{\Lambda}{r}}\left( \binom{r + \alpha}{r} - 1 \right)} \frac{\cancel{\binom{\Lambda}{r}}\binom{\Lambda}{\alpha}}{\binom{\Lambda}{r + \alpha}} \\
                                             & = \frac{\binom{\Lambda - \alpha}{r}}{\left( \binom{r + \alpha}{r} - 1 \right)} \frac{\binom{\Lambda}{\alpha}}{\binom{\Lambda}{r + \alpha}} \\
                                             & = \frac{\binom{r + \alpha}{r}}{\binom{r + \alpha}{r} - 1} \eqqcolon b_r
\end{align}
holds. Further, we notice that $b_r$ is decreasing in $r$, since
\begin{align}
    b_{r + 1} & = \frac{\binom{r + 1 + \alpha}{r + 1}}{\binom{r + 1 + \alpha}{r + 1} - 1} \\
              & = \frac{1}{1 - \frac{1}{\binom{r + 1 + \alpha}{r + 1}}} \\
              & = \frac{1}{1 - \frac{r + 1}{r + 1 + \alpha}\frac{1}{\binom{r + \alpha}{r}}} \\
              & < \frac{1}{1 - \frac{1}{\binom{r + \alpha}{r}}} \\
              & = b_{r}
\end{align}
for each $r \in \mathbb{N}^{+}$. Thus, considering that $r \in [\Lambda - \alpha + 1]$, we can further write
\begin{align}
    \frac{L_{\textnormal{UB}}(r)}{L^{\star}(r)} & \leq \frac{\binom{r + \alpha}{r}}{\binom{r + \alpha}{r} - 1} \\
                                              & \leq \frac{\alpha + 1}{\alpha}
\end{align}
where the last term is upper bounded when $\alpha$ is set to its minimum value. Now, after neglecting the value $\alpha = 1$ --- in which case the corresponding achievable performance in \Cref{thm: Achievable Bound Without Download Cost} was already proved to be exactly optimal in~\cite{Li2018FundamentalTradeoffComputation} --- we focus on the case where $\alpha \in [2 : \Lambda]$, which implies that
\begin{align}
    \frac{L_{\textnormal{UB}}(r)}{L^{\star}(r)} & \leq \frac{\alpha + 1}{\alpha} \\
                                              & \leq \frac{3}{2}.
\end{align}
The proof is concluded. \qed

\section{Proof of Order Optimality in \texorpdfstring{\Cref{thm: Order Optimality With Download Cost}}{Theorem~\ref{thm: Order Optimality With Download Cost}}}\label{sec: Order Optimality With Download Cost Proof}

From \Cref{thm: Achievable Bound With Download Cost} we know that $L^{\star}_{\text{max-link}}(r)$ is upper bounded as
\begin{align}
    L^{\star}_{\text{max-link}}(r) & \leq L_{\textnormal{max-link}, \textnormal{UB}}(r) \\ 
                                   & = \max\left(\sum_{j \in [\Lambda]} \frac{\binom{\Lambda - \alpha}{j}}{\binom{\Lambda}{j}\left(\binom{j + \alpha}{j} - 1\right)} \frac{\tilde{a}^{j}_{\star}}{N}, \sum_{j \in [\Lambda]} \frac{\binom{\Lambda}{\alpha} - \binom{\Lambda - j}{\alpha}}{\alpha \binom{\Lambda}{\alpha}} \frac{\tilde{a}^{j}_{\star}}{N} \right) \\
                                   & \leq \sum_{j \in [\Lambda]} \left( \frac{\binom{\Lambda - \alpha}{j}}{\binom{\Lambda}{j}\left(\binom{j + \alpha}{j} - 1\right)} + \frac{\binom{\Lambda}{\alpha} - \binom{\Lambda - j}{\alpha}}{\alpha \binom{\Lambda}{\alpha}} \right) \frac{\tilde{a}^{j}_{\star}}{N} \\
                                   & = \sum_{j \in [\Lambda]} c_{j} \frac{\tilde{a}^{j}_{\star}}{N}
\end{align}
where the coefficient $c_{j}$ is defined as
\begin{equation}
    c_{j} \coloneqq \frac{\binom{\Lambda - \alpha}{j}}{\binom{\Lambda}{j}\left(\binom{j + \alpha}{j} - 1\right)} + \frac{\binom{\Lambda}{\alpha} - \binom{\Lambda - j}{\alpha}}{\alpha \binom{\Lambda}{\alpha}}.
\end{equation}
At the same time, we know from \Cref{thm: Converse Bound With Download Cost} that $L^{\star}_{\text{max-link}}(r)$ is lower bounded as
\begin{align}
    L^{\star}_{\text{max-link}}(r) & \geq L_{\textnormal{max-link}, \textnormal{LB}}(r) \\
                                   & = \frac{1}{2} \sum_{j \in [\Lambda]} \left(\frac{\binom{\Lambda}{\alpha + j}}{\binom{\Lambda}{\alpha}\binom{\Lambda}{j}} + \frac{\binom{\Lambda}{\alpha} - \binom{\Lambda - j}{\alpha}}{\alpha \binom{\Lambda}{\alpha}} \right) \frac{\tilde{a}^{j}_{\star}}{N} \\
                                   & = \frac{1}{2} \sum_{j \in [\Lambda]} d_{j} \frac{\tilde{a}^{j}_{\star}}{N}
\end{align}
where the coefficient $d_{j}$ is defined as
\begin{equation}
    d_{j} \coloneqq \frac{\binom{\Lambda}{\alpha + j}}{\binom{\Lambda}{\alpha}\binom{\Lambda}{j}} + \frac{\binom{\Lambda}{\alpha} - \binom{\Lambda - j}{\alpha}}{\alpha \binom{\Lambda}{\alpha}}.
\end{equation}
Hence, we can evaluate the gap to optimality from the ratio $L_{\textnormal{max-link}, \textnormal{UB}}(r)/L_{\textnormal{max-link}, \textnormal{LB}}(r)$. In particular, we have
\begin{align}
    \frac{L_{\textnormal{max-link}, \textnormal{UB}}(r)}{L_{\textnormal{max-link}, \textnormal{LB}}(r)} & \leq 2\frac{\sum_{j \in [\Lambda]} c_{j} \tilde{a}^{j}_{\star}/N}{\sum_{j \in [\Lambda]} d_{j} \tilde{a}^{j}_{\star}/N} \\ 
    & = 2 \frac{\sum_{j \in [\Lambda]: \tilde{a}^{j}_{\star} > 0} c_{j} \tilde{a}^{j}_{\star}/N}{\sum_{j \in [\Lambda]: \tilde{a}^{j}_{\star} > 0} d_{j} \tilde{a}^{j}_{\star}/N} \\
    & \leq 2 \max_{j \in [\Lambda]: \tilde{a}^{j}_{\star} > 0} \frac{c_{j} \tilde{a}^{j}_{\star}/N}{d_{j} \tilde{a}^{j}_{\star}/N} \\
    & = 2 \max_{j \in [\Lambda]: \tilde{a}^{j}_{\star} > 0 } \frac{c_{j}}{d_{j}} \\
    & \leq 2 \max_{j \in [\Lambda]} \frac{c_{j}}{d_{j}} \\
    & = 2 \max \left(\max_{j \in [\Lambda - \alpha]} \frac{c_{j}}{d_{j}}, \max_{j \in [\Lambda - \alpha + 1 : \Lambda]} \frac{c_{j}}{d_{j}} \right).
\end{align}
Now, we can see that $c_{j} = d_{j}$ when $j > \Lambda - \alpha$. Else, when $j \in [\Lambda - \alpha]$, we have
\begin{align}
    \frac{c_{j}}{d_{j}} & = \frac{\frac{\binom{\Lambda - \alpha}{j}}{\binom{\Lambda}{j}\left(\binom{j + \alpha}{j} - 1\right)} + \frac{\binom{\Lambda}{\alpha} - \binom{\Lambda - j}{\alpha}}{\alpha \binom{\Lambda}{\alpha}}}{\frac{\binom{\Lambda}{\alpha + j}}{\binom{\Lambda}{\alpha}\binom{\Lambda}{j}} + \frac{\binom{\Lambda}{\alpha} - \binom{\Lambda - j}{\alpha}}{\alpha \binom{\Lambda}{\alpha}}} \\
    & \leq \max \left(\frac{\frac{\binom{\Lambda - \alpha}{j}}{\binom{\Lambda}{j}\left(\binom{j + \alpha}{j} - 1\right)}}{\frac{\binom{\Lambda}{\alpha + j}}{\binom{\Lambda}{\alpha}\binom{\Lambda}{j}}}, \frac{\frac{\binom{\Lambda}{\alpha} - \binom{\Lambda - j}{\alpha}}{\alpha \binom{\Lambda}{\alpha}}}{\frac{\binom{\Lambda}{\alpha} - \binom{\Lambda - j}{\alpha}}{\alpha \binom{\Lambda}{\alpha}}}  \right) \\
    & = \max \left(\frac{\binom{\Lambda - \alpha}{j}}{\left(\binom{j + \alpha}{j} - 1\right)}\frac{\binom{\Lambda}{\alpha}}{\binom{\Lambda}{\alpha + j}}, 1\right).
\end{align}
Since we know from \refappendix{sec: Order Optimality Without Download Cost Proof} that
\begin{equation}
    \frac{\binom{\Lambda - \alpha}{j}}{\left(\binom{j + \alpha}{j} - 1\right)}\frac{\binom{\Lambda}{\alpha}}{\binom{\Lambda}{\alpha + j}} \leq \frac{\alpha + 1}{\alpha}
\end{equation}
we can further write
\begin{align}
    \max_{j \in [\Lambda - \alpha]} \frac{c_{j}}{d_{j}} & \leq \max_{j \in [\Lambda - \alpha]} \max \left(\frac{\binom{\Lambda - \alpha}{j}}{\left(\binom{j + \alpha}{j} - 1\right)}\frac{\binom{\Lambda}{\alpha}}{\binom{\Lambda}{\alpha + j}}, 1\right)  \\
    & \leq \max \left( \frac{\alpha + 1}{\alpha}, 1 \right).
\end{align}
Hence, we can conclude that
\begin{align}
    \frac{L_{\textnormal{max-link}, \textnormal{UB}}(r)}{L_{\textnormal{max-link}, \textnormal{LB}}(r)} & \leq 2 \max \left(\max_{j \in [\Lambda - \alpha]} \frac{c_{j}}{d_{j}}, \max_{j \in [\Lambda - \alpha + 1 : \Lambda]} \frac{c_{j}}{d_{j}} \right) \\
    & \leq 2 \max \left( \max \left( \frac{\alpha + 1}{\alpha}, 1\right) , 1\right) \\
    & \leq 2 \max \left( \max \left( 2, 1\right) , 1\right) \\
    & = 4.
\end{align}
The proof is concluded. \qed

\section{Proof of \texorpdfstring{\Cref{lem: Lower Bound Lemma}}{Lemma~\ref{lem: Lower Bound Lemma}}}\label{sec: Lower Bound Lemma Proof}

Consider a permutation $\bm{c} = (c_1, \dots, c_{\Lambda})$ of the set $[\Lambda]$. We know that $H(D_{\mathcal{U}} \mid X_{[\Lambda]_\alpha}, C_{\mathcal{U}}) = 0$ holds for any valid shuffle scheme and for each $\mathcal{U} \in [\Lambda]_{\alpha}$. Given this, for $\mathcal{U}^{\alpha} = \{c_1, \dots, c_{\alpha}\}$ we can write
\begin{align}
    H(X_{[\Lambda]_\alpha}) & \geq H(X_{[\Lambda]_\alpha} \mid C_{\mathcal{U}^{\alpha}}) \label{eqn: Conditioning 1} \\
             & =  H(X_{[\Lambda]_\alpha}, D_{\mathcal{U}^{\alpha}} \mid C_{\mathcal{U}^{\alpha}}) - H(D_{\mathcal{U}^{\alpha}} \mid X_{[\Lambda]_\alpha}, C_{\mathcal{U}^{\alpha}}) \\
             & = H(X_{[\Lambda]_\alpha}, D_{\mathcal{U}^{\alpha}} \mid C_{\mathcal{U}^{\alpha}}) \label{eqn: Decodability 1} \\
             & = H(D_{\mathcal{U}^{\alpha}} \mid C_{\mathcal{U}^{\alpha}}) + H(X_{[\Lambda]_\alpha} \mid C_{\mathcal{U}^{\alpha}}, D_{\mathcal{U}^{\alpha}}) \\
             & = H(\mathcal{D}_{\mathcal{U}^{\alpha}} \mid \mathcal{C}_{\mathcal{U}^{\alpha}}) + H(X_{[\Lambda]_\alpha} \mid \mathcal{Y}_{\alpha})
\end{align}
where \eqref{eqn: Conditioning 1} follows from the fact that conditioning does not increase entropy, and where \eqref{eqn: Decodability 1} holds because of the decodability condition $H(D_{\mathcal{U}} \mid X_{[\Lambda]_\alpha}, C_{\mathcal{U}}) = 0$ for each $\mathcal{U} \in [\Lambda]_{\alpha}$. Similarly, for each $i \in [\alpha + 1 : \Lambda]$ we can write
\begin{align}
    H(X_{[\Lambda]_\alpha} \mid \mathcal{Y}_{i - 1}) & \geq H(X_{[\Lambda]_\alpha} \mid \mathcal{C}_{i}, \mathcal{Y}_{i - 1}) \label{eqn: Conditioning 2} \\
                                   & = H(X_{[\Lambda]_\alpha}, \mathcal{D}_i \mid \mathcal{C}_{i}, \mathcal{Y}_{i - 1}) - H(\mathcal{D}_i \mid X_{[\Lambda]_\alpha}, \mathcal{C}_{i}, \mathcal{Y}_{i - 1}) \\
                                   & = H(X_{[\Lambda]_\alpha}, \mathcal{D}_i \mid \mathcal{C}_{i}, \mathcal{Y}_{i - 1}) \label{eqn: Decodability 2} \\
                                   & = H(\mathcal{D}_i \mid \mathcal{C}_{i}, \mathcal{Y}_{i - 1}) + H(X_{[\Lambda]_\alpha} \mid \mathcal{D}_i, \mathcal{C}_{i}, \mathcal{Y}_{i - 1}) \\
                                   & = H(\mathcal{D}_i \mid \mathcal{C}_{i}, \mathcal{Y}_{i - 1}) + H(X_{[\Lambda]_\alpha} \mid \mathcal{Y}_{i})
\end{align}
where again \eqref{eqn: Conditioning 2} is true as conditioning does not increase entropy, and where \eqref{eqn: Decodability 2} follows because
\begin{align}
    H(\mathcal{D}_i \mid X_{[\Lambda]_\alpha}, \mathcal{C}_{i}, \mathcal{Y}_{i - 1}) & \leq H(\mathcal{D}_i \mid X_{[\Lambda]_\alpha}, \mathcal{C}_{i}) \\
    & \leq \sum_{\mathcal{U}^i \subseteq \{c_1, \dots, c_i\} : |\mathcal{U}^i| = \alpha, c_i \in \mathcal{U}^i} H(D_{\mathcal{U}^i} \mid X_{[\Lambda]_\alpha}, C_{\mathcal{U}^i}) \\
    & = 0
\end{align}
due to the independence of intermediate values and the decodability condition. Considering that $H(X_{[\Lambda]_\alpha} \mid \mathcal{Y}_{\Lambda}) = 0$, we can use iteratively the above to obtain
\begin{equation}
    H(X_{[\Lambda]_\alpha}) \geq \sum_{i \in [\alpha : \Lambda]} H(\mathcal{D}_{i} \mid \mathcal{C}_{i}, \mathcal{Y}_{i - 1}).
\end{equation}
Further, we notice that $L_{\mathcal{M}} \geq H(X_{[\Lambda]_\alpha})/QNT$. This concludes the proof. \qed

\section{Proof of \texorpdfstring{\Cref{lem: Download Cost Lemma}}{Lemma~\ref{lem: Download Cost Lemma}}}\label{sec: Download Cost Lemma Proof}

First, we rewrite the equality in \Cref{lem: Download Cost Lemma} as
\begin{align}
    \sum_{i \in [\alpha]} \frac{\binom{\Lambda - \alpha}{j - i} \binom{\alpha - 1}{i - 1}}{i\binom{\Lambda}{j}} & = \frac{\binom{\Lambda}{j} - \binom{\Lambda - \alpha}{j}}{\alpha \binom{\Lambda}{j}} \\
    \sum_{i \in [\alpha]} \binom{\Lambda - \alpha}{j - i} \binom{\alpha}{i} & = \binom{\Lambda}{j} - \binom{\Lambda - \alpha}{j} \\
    \sum_{i \in [0 : \alpha]} \binom{\Lambda - \alpha}{j - i} \binom{\alpha}{i} & = \binom{\Lambda}{j}.
\end{align}
Thus, proving the equality in \Cref{lem: Download Cost Lemma} is equivalent to showing that the following equality
\begin{equation}
    \sum_{i \in [0 : \alpha]} \binom{\Lambda - \alpha}{j - i} \binom{\alpha}{i} = \binom{\Lambda}{j}
\end{equation}
holds. From Vandermonde's identity, we know that 
\begin{equation}
    \sum_{i \in [0 : j]} \binom{\Lambda - \alpha}{j - i} \binom{\alpha}{i} = \binom{\Lambda}{j}
\end{equation}
and so it suffices to show that
\begin{equation}
    \sum_{i \in [0 : \alpha]} \binom{\Lambda - \alpha}{j - i} \binom{\alpha}{i} = \sum_{i \in [0 : j]} \binom{\Lambda - \alpha}{j - i} \binom{\alpha}{i}.
\end{equation}

Consider first the case $j \leq \alpha$. This means that we can write
\begin{align}
    \sum_{i \in [0 : \alpha]} \binom{\Lambda - \alpha}{j - i} \binom{\alpha}{i} & = \sum_{i \in [0 : j]} \binom{\Lambda - \alpha}{j - i} \binom{\alpha}{i} + \underbrace{\sum_{i \in [j + 1 : \alpha]} \binom{\Lambda - \alpha}{j - i} \binom{\alpha}{i}}_{= 0} \\
    & = \sum_{i \in [0 : j]} \binom{\Lambda - \alpha}{j - i} \binom{\alpha}{i}
\end{align}
where $\sum_{i \in [j + 1 : \alpha]} \binom{\Lambda - \alpha}{j - i} \binom{\alpha}{i} = 0$ since we have $\binom{\Lambda - \alpha}{j - i} = 0$ for $i \in [j + 1 : \alpha]$. Similarly, if we consider $j \geq \alpha$, we have
\begin{align}
    \sum_{i \in [0 : j]} \binom{\Lambda - \alpha}{j - i} \binom{\alpha}{i} & = \sum_{i \in [0 : \alpha]} \binom{\Lambda - \alpha}{j - i} \binom{\alpha}{i} + \underbrace{\sum_{i \in [\alpha + 1 : j]} \binom{\Lambda - \alpha}{j - i} \binom{\alpha}{i}}_{= 0} \\
    & = \sum_{i \in [0 : \alpha]} \binom{\Lambda - \alpha}{j - i} \binom{\alpha}{i}
\end{align}
where $\sum_{i \in [\alpha + 1 : j]} \binom{\Lambda - \alpha}{j - i} \binom{\alpha}{i} = 0$ since we have $\binom{\alpha}{i} = 0$ for $i \in [\alpha + 1 : j]$. Hence, for any value of $j \in [0 : \Lambda]$, we can conclude that
\begin{equation}
    \sum_{i \in [0 : \alpha]} \binom{\Lambda - \alpha}{j - i} \binom{\alpha}{i} = \sum_{i \in [0 : j]} \binom{\Lambda - \alpha}{j - i} \binom{\alpha}{i} = \binom{\Lambda}{j}.
\end{equation}
The proof is concluded. \qed

\bibliographystyle{IEEEtran}
\bibliography{references}

\end{document}